\documentclass[conference]{IEEEtran}
\IEEEoverridecommandlockouts
\usepackage{amsmath,amssymb,amsfonts, mathtools,algorithmic,graphicx,textcomp,xcolor}
\def\BibTeX{{\rm B\kern-.05em{\sc i\kern-.025em b}\kern-.08em
    T\kern-.1667em\lower.7ex\hbox{E}\kern-.125emX}}
\usepackage{amsthm, xpatch, hyperref, pgfplots, tikz, forest, epstopdf, pdftexcmds, floatrow, subcaption, caption, subfiles, csquotes} 
\usepackage[sort, compress]{cite}
\usepackage{array,makecell}

\tikzset{
	dot/.style={circle,draw,inner sep=1.2,fill=black},
}

\usepackage[utf8]{inputenc}
\usepackage[english]{babel}
\makeatletter
\xpatchcmd{\proof}{\topsep6\p@\@plus6\p@\relax}{}{}{}
\makeatother

\newtheorem{env_example}{Example}
\newtheorem{defn}{Definition}
\newtheorem{theorem}{Theorem}

\newtheorem{lemma}{Lemma}

\DeclarePairedDelimiter{\ceil}{\lceil}{\rceil}
\DeclarePairedDelimiter{\floor}{\lfloor}{\rfloor}
\graphicspath{{./figures/}}
\allowdisplaybreaks

\newcounter{cases}
\newcounter{subcases}[cases]
\newenvironment{mycases}
{%
	\setcounter{cases}{0}%
	\setcounter{subcases}{0}%
	\def\case
	{%
		\par\noindent
		\refstepcounter{cases}%
		\textbf{Case \thecases.}
	}%
	\def\subcase
	{%
		\par\noindent
		\refstepcounter{subcases}%
		\textit{Subcase (\thesubcases):}
	}%
}
{%
	\par
}
\renewcommand*\thecases{\arabic{cases}}
\renewcommand*\thesubcases{\roman{subcases}}
\usetikzlibrary{matrix,backgrounds}

\begin{document}

\title{Insertion and Deletion Correction in\\ Polymer-based Data Storage\\
\thanks{
A. Banerjee and A. Wachter-Zeh are with the Institute for Communications Engineering, Technical University of Munich, DE-80333, Munich, Germany. E-mails: \{anisha.banerjee,antonia.wachter-zeh\}@tum.de.

E. Yaakobi is with the Computer Science Department, Technion–Israel Institute of  Technology,  Haifa  32000,  Israel. E-mail: yaakobi@cs.technion.ac.il.

This work has been supported by the European Research
Council (ERC) under the European Union’s Horizon 2020 research and innovation programme (Grant Agreement No. 801434).}
}

\author{Anisha Banerjee, Antonia Wachter-Zeh and Eitan Yaakobi
}

\maketitle
\thispagestyle{plain}
\pagestyle{plain}

\begin{abstract}
Synthetic polymer-based storage seems to be a particularly promising candidate that could help to cope with the ever-increasing demand for archival storage requirements. It involves designing molecules of distinct masses to represent the respective bits $\{0,1\}$, followed by the synthesis of a polymer of molecular units that reflects the order of bits in the information string. Reading out the stored data requires the use of a tandem mass spectrometer, that fragments the polymer into shorter substrings and provides their corresponding masses, from which the \emph{composition}, i.e. the number of $1$s and $0$s in the concerned substring can be inferred. Prior works have dealt with the problem of unique string reconstruction from the set of all possible compositions, called \emph{composition multiset}. This was accomplished either by determining which string lengths always allow unique reconstruction, or by formulating coding constraints to facilitate the same for all string lengths. Additionally, error-correcting schemes to deal with substitution errors caused by imprecise fragmentation during the readout process, have also been suggested. This work builds on this research by generalizing previously considered error models, mainly confined to substitution of compositions. To this end, we define new error models that consider insertions of spurious compositions and deletions of existing ones, thereby corrupting the composition multiset. We analyze if the reconstruction codebook proposed by Pattabiraman \emph{et al.} is indeed robust to such errors, and if not, propose new coding constraints to remedy this.
\end{abstract}


\section{Introduction}

As we progress through this digital age, our rate of data generation continues to rise unhindered, and with it, so do our storage requirements. Since current data storage media are not particularly advantageous in regard to longevity or density, several molecular storage techniques \cite{mol1, mol2, mol3, mol3a, mol3b, mol4, mol5, mol6, mol7} have been proposed. The work in \cite{mol1} involving synthetic polymer-based storage systems appears to be especially favorable, given its promise of efficient synthesis, low read latency and cost. Under this paradigm, a string of information bits is encoded into a chain of molecules linked by means of phosphate bonds, such that the component molecules may only assume one of two significantly differing masses, which represent the bits $0$ and $1$ respectively. The stored data can be read out by employing a tandem mass (MS/MS) spectrometer, which essentially splits the synthesized polymer at the phosphate linkages and outputs the masses of the resulting fragments. In this manner, the user is given access to the masses of all substrings in the encoded string.

The previous work \cite{acharya} dealt with the problem of reconstructing a binary string from such an MS/MS readout, under the following modeling assumptions:

 \emph{Assumption 1.} Masses of the component molecules are chosen such that one can always uniquely infer the \emph{composition}, i.e., the number of $0$s and $1$s forming a certain fragment, from its mass.
 
 \emph{Assumption 2.} While fragmenting a polymer for the purpose of mass spectrometry analysis, the masses of all constituent substrings are observed with identical frequency.
 
This proposed setting simplifies the recovery of the original information string into the problem of binary string reconstruction from its composition multiset. More specifically, the reconstruction process now involves determining the binary string from a set of compositions of all of its substrings of each possible length.
It is worth noting that this setup does not allow for differentiation between a string and its reversal, since their sets of substring compositions would be identical.

While the authors of \cite{acharya} primarily focused on string lengths that ensured unique reconstruction from a composition multiset, subsequent works \cite{pattabiraman, p2, p3} extended this research by building a code that allows for unique reconstruction of each member codeword from its composition multiset alone, regardless of the string length. It was found that a redundancy proportional to the logarithm of the information length is sufficient to guarantee unique reconstruction. Similar coding constraints were also proposed to also cope with possible errors in the composition multiset. The work in \cite{p4} takes a step further by dealing with the recovery of multiple strings from the mass spectrometry readout of a mixture of synthesized polymers.

Since the errors introduced during an MS/MS readout are often context-dependent, we devote this work to the generalization of the error model considered in \cite{pattabiraman,p2}. Specifically, we investigate the impact of inserting and deleting one or more compositions on the reconstructability of the encoded strings. In addition to this, new coding constraints are proposed to enable the correction of such errors. We also consider a special kind of substitution error, namely a \emph{skewed substitution error}. This category of errors is motivated by imperfect fragmentations of a given polymer during the MS/MS readout process, as a result of which the observed molecular mass of a shorter monomer chain is lower than what the true mass of its perfectly fragmented version would have been. In this scenario, errors occur only in one direction, i.e., the the measured mass can only be lower than the true mass, not higher. An error-correcting scheme is also suggested for this setting.

The organization of this work is as follows. Section \ref{sec::prelim} introduces relevant terminology, notations and some preliminary results to be exploited subsequently. Section \ref{sec::subs} discusses coding constructions proposed in earlier works \cite{pattabiraman, p2,p3}, while Section \ref{sec::error_models} describes the error models pertaining to insertions, deletions and skewed substitutions of one or multiple compositions and also briefly summarizes error-correcting codes to deal with the same. We demonstrate the equivalence between codes correcting deletions and insertions of multisets in Section \ref{sec::equiv}. Sections \ref{sec::adel} and \ref{sec::sdel} delve deeper into the constructions capable of correcting deletions of multiple multisets. We also talk about skewed substitution errors and related coding constructions in Section \ref{sec::skew}. Finally, we conclude with Section \ref{sec::concl}, where a few open problems are discussed.

\section{Preliminaries} \label{sec::prelim}

Let $\boldsymbol{s}=s_1s_2\ldots s_n$ denote a binary string of $n$ bits. Any substring $s_i \ldots s_j$ where $i \leq j$, may be indicated by $\boldsymbol{s}_i^j$. The \emph{composition} of this substring, denoted by $c(\boldsymbol{s}_i^j)$, is said to be $0^z 1^w$, where $z$ and $w$ refer to the number of $0$s and $1$s in $\boldsymbol{s}_i^j$ respectively, such that $z+w=j-i+1$. We also define $C_k (\boldsymbol{s})$ as the set of compositions of all length-$k$ substrings in $\boldsymbol{s}$. Evidently, $C_k(\boldsymbol{s})$ should contain $n-k+1$ compositions.
\begin{env_example}
	Consider $\boldsymbol{s}=001010111$. Then, the multiset of compositions for substrings of length $7$ is given by: $C_7(\boldsymbol{s})=\{0^41^3,0^31^4,0^21^5\}$.
\end{env_example} 
Upon combining the multisets for all $1\leq k \leq n$, we obtain the \emph{composition multiset} of $\boldsymbol{s}$:
\begin{equation}
	C(\boldsymbol{s})=\bigcup_{k\in [n]}C_k(\boldsymbol{s}). \nonumber
\end{equation}
where $[n]=\{1, \ldots, n\}$. 
As stated earlier, \cite{acharya} determined string lengths for which unique reconstruction (up to reversal) from such sets is possible. For the remaining string lengths, the authors exploited a bivariate generating polynomial representation, to find strings that are equicomposable with a given string. Here, two distinct strings $\boldsymbol{s}, \boldsymbol{t} \in \{0,1\}^n$ are said to be \emph{equicomposable} if a common composition multiset is shared, i.e., $C(\boldsymbol{s})=C(\boldsymbol{t})$.

A code $\mathcal{C}$ is called a \emph{composition-reconstructable code} if for all $\boldsymbol{s}, \boldsymbol{t}\in \mathcal{C}$, it holds that $C(\boldsymbol{s})\neq C(\boldsymbol{t})$. For all $n$, denote by $A(n)$ the size of the largest composition reconstructable code. Since composition multisets are identical for a binary string and its reversal, it holds that $$A(n)\leq 2^{\ceil{\frac{n}{2}}}+\frac{1}{2}(2^n - 2^{\ceil{\frac{n}{2}}}) = 2^{n-1}+2^{\ceil{\frac{n}{2}}-1},$$ 
where the term $2^{\ceil{\frac{n}{2}}}$ describes the number of palindromic strings of length $n$, and \cite{acharya} determined string lengths $n$ where it is possible to achieve this bound with equality. Specifically, it was shown that binary strings of length $\leq 7$, one less than a prime, or one less than twice a prime, are uniquely reconstructable up to reversal.

\subsection{Unique Reconstruction Codes}
For values of $n$ where it is not possible to achieve the aforementioned bound, it is necessary to formulate a code, as done in \cite{pattabiraman, p2}.

The first major coding-theoretic problem concerning polymer-based storage involved designing constraints in order to guarantee unique reconstruction for codewords of a fixed length, i.e., to formulate a composition-reconstructable code. To this end, \cite{ p2} introduced the following {composition-reconstructable code} for even codeword lengths.

\emph{Construction 1 \cite{p2}:} 
\begin{equation}
\begin{split}
\mathcal{S}_R(n)=\
&\{\boldsymbol{s}\in \{0,1\}^n, s_1=0, s_n=1, \text{ and} \\
& \exists I \subset \{2,\ldots, n-1\} \text{ such that}\\
& \quad \quad \quad \quad \text{ for all } i \in I, s_i\neq s_{n+1-i}, \\
& \quad \quad \quad \quad \text{ for all } i \notin I, s_i=s_{n+1-i},\\
& \quad \quad \boldsymbol{s}_{[n/2]\cap I} \text{ is a Catalan-Bertrand string.}\}\\
\end{split}\label{eq::sr}
\end{equation}
In this context, a Catalan-Bertrand string refers to any binary vector wherein each prefix contains strictly more $0$s than $1$s. When $n$ is odd, the codebook $\mathcal{S}_R(n)$ is defined as: 
\begin{equation}
	\mathcal{S}_R(n)=\hspace{-2ex}\bigcup_{\boldsymbol{s} \in \mathcal{S}_R(n-1)} \hspace{-2ex} \{\boldsymbol{s}_1^{(n-1)/2}0\boldsymbol{s}_{(n+1)/2}\}^n, \boldsymbol{s}_1^{(n-1)/2}1\boldsymbol{s}_{(n+1)/2}^n\}. \label{eq:sr_odd}
\end{equation}

The number of redundant bits can thus be upper-bounded in terms of $n$ as $1/2 \log(n) +5$ \cite{pattabiraman}. Alternatively, we obtain the following statement from \cite{p2}. 
\begin{theorem}
    {\cite[pg. 3]{p2}}
     There exist efficiently encodable and decodable reconstruction codes with $k$ information bits and redundancy at most $\frac{1}{2}\log(k)+6$.
\end{theorem}

From the definition of $A(n)$, we can also deduce that,
\begin{equation}
|\mathcal{S}_R(n)| \leq A(n). \nonumber    
\end{equation}

This construction sets $s_1=0$ and $s_n=1$ to avoid confusion among reversals, while the remaining bits are chosen such that the weight of a prefix and a suffix of equal length are unequal if the said prefix includes a Catalan-Bertrand string, i.e.,
\begin{equation}
\mathrm{wt}(\boldsymbol{s}_2^i) \begin{cases}
=\mathrm{wt}(\boldsymbol{s}_{n-i+1}^{n-1}), & \text{if } [i] \cap I=\emptyset,\\
<\mathrm{wt}(\boldsymbol{s}_{n-i+1}^{n-1}), & \text{otherwise},
\end{cases} \label{ineq::imp}
\end{equation}
where $i<\ceil{\frac{n}{2}}$ and $\mathrm{wt}(\cdot)$ denotes the Hamming weight of the argument. 
The latter inequality stems from the fact that if $\boldsymbol{s}_{[i]\cap I}$ has strictly more $0$s than $1$s, then $\boldsymbol{s}_{\{n-i+1, \ldots, n-1\}\cap I}$ contains strictly more $1$s than $0$s, thus causing a weight mismatch. Here, we note that the embedded Catalan-Bertrand string may begin from index 2 at the earliest.

\begin{subsection}{Reconstruction from Error-Free Composition Multisets} \label{subsec::rec}
The decoder of the composition-reconstructable code $\mathcal{S}_R(n)$ recovers a string from its composition multiset by employing the approach outlined in \cite{acharya, pattabiraman}. Since the underlying principles of this process help us in formulating coding constructions for the more general error models involving insertions and deletions, we briefly discuss it in this subsection. For further details, the reader is referred to \cite{acharya, pattabiraman}.

The algorithm begins by deducing the following sequence that characterizes the string to be recovered, say $\boldsymbol{s} \in \mathcal{S}_R(n)$, 
$$\boldsymbol{\sigma}_{s}=(\sigma_1, \ldots, \sigma_{\ceil{n/2}}),$$
where $\sigma_i=\mathrm{wt}(s_is_{n-i+1})$ for $i \in \{1,\ldots, \floor{n/2}\}$. When $n$ is odd, we set $\sigma_{\ceil{\frac{n}{2}}}=\mathrm{wt}(s_{\ceil{\frac{n}{2}}})$, i.e., the weight of the central element.
\begin{env_example}
	For $\boldsymbol{s}=001010111$. the sequence of $\sigma_i$'s is $\boldsymbol{\sigma}_{s}=(1,1,2,0,1)$.
\end{env_example}
These values can be computed by exploiting some inherent properties of composition multisets. In particular, we make use of \emph{cumulative weights}, which are defined for each multiset $C_k(\boldsymbol{s})$ as: $$w_k(\boldsymbol{s})=\sum_{0^z1^w \in C_k(\boldsymbol{s})}w.$$ 
\begin{env_example}
	For instance, the multiset $C_7(\boldsymbol{s})=\{0^41^3,0^31^4,0^21^5\}$ has a cumulative weight $w_7(\boldsymbol{s})=12$.
\end{env_example}
It is easy to see that for all $k\leq \ceil{\frac{n}{2}}$, these weights obey the following relations:
\begin{align}
w_1(\boldsymbol{s})&=\sum_{i=1}^{\ceil{\frac{n}{2}}} \sigma_i, \label{eq::w1} \\
w_k(\boldsymbol{s})&=\sum_{i=1}^{k}i\sigma_i+k\sum_{i=k+1}^{\ceil{n/2}}\sigma_{i}\label{eq::wk_alt}\\
&=kw_1(\boldsymbol{s})-\sum_{i=1}^{k-1}i \sigma_{k-i}.  \label{eq::cum_wts}
\end{align}

 We also observe a symmetry relation for any given set of cumulative weights: 
\begin{equation}
w_k(\boldsymbol{s})=w_{n-k+1}(\boldsymbol{s}), \quad \forall \; k \in [n]. \label{eq::wt_sym}
\end{equation}
In light of this, the multisets $C_i$ and  $C_{n-i+1}$ are henceforth said to be \emph{symmetric}. For notational convenience, we also define:
$$\widetilde{C}_i(\boldsymbol{s})=C_i(\boldsymbol{s}) \cup C_{n-i+1}(\boldsymbol{s}) $$
Now to demonstrate the functioning of the reconstruction algorithm, we consider the following example.
\begin{env_example}
	In this example, we reconstruct the string $\boldsymbol{s}=001010111$ from its composition multiset $C(\boldsymbol{s})$, which is stated below:
	\begin{equation}
		\begin{split}
		C(\boldsymbol{s})=\{& 0,0,1,0,1,0,1,1,1, 0^2,0^11^1,0^11^1,0^11^1,0^11^1,\\
		&0^11^1,1^2,1^2,0^21^1,0^21^1,0^11^2,0^21^1,0^11^2,0^11^2,\\
		&1^3,0^31^1,0^21^2,0^21^2,0^21^2,0^11^3,0^11^3,0^31^2,\\
		&0^31^2,0^21^3,0^21^3,0^11^4,0^41^2,0^31^3,0^21^4,0^21^4,\\
		&0^41^3,0^31^4,0^21^5,0^41^4,0^31^5,0^41^5\}.
		\end{split}
	\end{equation}
	The reconstruction process involves the following steps:
	\begin{enumerate}
		\item Firstly, we deduce its $\boldsymbol{\sigma}_{s}$ sequence from (\ref{eq::w1}) and (\ref{eq::cum_wts}):
		$$\boldsymbol{\sigma}_{s}=(1,1,2,0,1).$$
		\item We create a multiset $\mathcal{T}$ to include all compositions that can be determined from $\boldsymbol{\sigma}_{s}$. More explicitly, one can infer the compositions $c(s_{5}), c(\boldsymbol{s}_{4}^6), \ldots, c(\boldsymbol{s}_{1}^9)$ by noting that for any $i<\ceil{n/2}$,
		\begin{equation}
			c(s_is_{n-i+1})=\begin{cases}
			0^2, & \text{if } \sigma_i=0.\\
			0^11^1, & \text{if } \sigma_i=1.\\
			1^2, & \text{if } \sigma_i=2.
			\end{cases} \nonumber
		\end{equation}
		 $$\mathcal{T}=\{1, 0^21, 0^21^3, 0^31^4, 0^41^5\}.$$
		
		\item The process now assigns the bits of $\boldsymbol{s}$ pairwise, in an inward manner, starting with bit pair $(s_1,s_9)$. Since $\sigma_1=1$, we could set $s_1=0$ and $s_9=1$ or vice-versa. Due to (\ref{eq::sr}), we opt for the former, i.e. $(s_1,s_9)=(0,1)$.
		\item Using the reconstructed prefix and suffix, we update $\mathcal{T}$:
		 $$\mathcal{T}=\{0,1,1, 0^21, 0^21^3, 0^31^4, 0^41^5,0^31^5,0^41^4\}.$$
		\item The two longest compositions in the multiset $C(\boldsymbol{s}) \backslash \mathcal{T}$ are $\{0^41^3, 0^21^5\}$. These denote the compositions of substrings $\boldsymbol{s}_1^7$ and $\boldsymbol{s}_3^9$. Conversely, their complements $\{1^2, 0^2\}$ correspond to substrings $\boldsymbol{s}_1^2$ and $\boldsymbol{s}_8^9$. Combining this with the knowledge of bits $s_1$ and $s_9$, we reconstruct $\boldsymbol{s}$ up to its prefix-suffix pair of length 2, i.e. $(\boldsymbol{s}_1^2,\boldsymbol{s}_8^9)=(00,11)$.
		\item To recover the remaining bits, we simply repeat steps 4 and 5.
	\end{enumerate}
\end{env_example}
\end{subsection}
\section{Substitution-Correcting Constructions} \label{sec::subs}
We now turn our attention to the problem of reconstruction from erroneous composition multisets. Substitution errors were considered in \cite{pattabiraman} under the asymmetric and symmetric setting. In this error model, some compositions in $C(\boldsymbol{s})$ are arbitrarily altered. If the errors occur such that each multiset $\widetilde{C}_i$ includes at most one substituted composition, then they are said to be \emph{asymmetric}. On the contrary, a pair of \emph{symmetric} substitution errors would occur in the multisets $C_i$ and $C_{n-i+1}$, for any $i \in [n]$. 
\begin{defn}
	A composition multiset $C(\boldsymbol{s})$ of the string $\boldsymbol{s} \in \{0,1,\}^n$ is said to have suffered an \textbf{asymmetric substitution error}, if for some $i \in [n]$, a single composition of the multiset $C_i(\boldsymbol{s})$ is modified, but its symmetric counterpart $C_{n-i+1}(\boldsymbol{s})$ remains unaffected.
\end{defn}
\begin{defn}
	If a composition multiset $C(\boldsymbol{s})$ is corrupted by having one composition substituted in each of the multisets $C_i(\boldsymbol{s})$ and $C_{n-i+1}(\boldsymbol{s})$, then \textbf{two symmetric substitution errors} are said to have occurred.
\end{defn}
To exemplify this, we consider the following.
\begin{env_example}
	Let $\boldsymbol{s}=001010111$. The symmetric multiset pair $C_3(\boldsymbol{s})$ and $C_7(\boldsymbol{s})$ is given by
	\begin{align}
	    C_3(\boldsymbol{s})&= \{0^21, 0^21, 01^2, 0^21, 01^2, 01^2, 1^3\}, \nonumber \\
	    C_7(\boldsymbol{s})&=\{0^41^3,0^31^4,0^21^5\}. \nonumber
	\end{align}
	  For instance, an asymmetric substitution error is said to have occurred if $C_7(\boldsymbol{s})$ is corrupted to
	\begin{equation}
	    C'_7(\boldsymbol{s})=\{0^41^3,0^31^4,0^31^4\}. \nonumber
	\end{equation}
	On the contrary, if $C_3(\boldsymbol{s})$ is also corrupted in addition to $C_7(\boldsymbol{s})$ as follows,
	\begin{align}
	    C'_3(\boldsymbol{s})&= \{1^3, 0^21, 01^2, 0^21, 01^2, 01^2, 1^3\}, \nonumber
	\end{align}
	then two symmetric substitution errors are said to have occurred.
\end{env_example}
We recall an important construction from \cite{pattabiraman} that corrects such composition substitution errors. In the following, we designate a code $\mathcal{S}_{CA}^{(t)}$ as a \emph{$t$-asymmetric composition code}, if for all $\boldsymbol{s}$, $\boldsymbol{v} \in \mathcal{S}_{CA}^{(t)}$, there exists no $\mathcal{I} \subseteq [\ceil{\frac{n}{2}}]$ with $|\mathcal{I}| \leq t$ such that 
\begin{align}
    |\widetilde{C}_i(\boldsymbol{s}) \setminus \widetilde{C}_i(\boldsymbol{v}|&=1 \quad \forall \; i \in \mathcal{I}, \nonumber \\
    \widetilde{C}_i(\boldsymbol{s})&= \widetilde{C}_i(\boldsymbol{v}) \quad \forall \; i \in \Big[\Big\lceil\frac{n}{2}\Big\rceil\Big]\setminus \mathcal{I}. \nonumber
\end{align}

\emph{Construction 2 \cite{pattabiraman,p2}:} A single (asymmetric or symmetric) composition code for odd values of $n$ is stated below.
\begin{equation}
	\begin{split}
	\mathcal{S}^{(1)}_{CA}(n)=
	&\{\boldsymbol{s}=s_1 s^*_1 s_2 \ldots s_{\ceil{\frac{n-2}{2}}} \ldots s_{n-3}s^*_n s_{n-2}\in \{0,1\}^n:\\
	& s_1 \ldots s_{n-2} \in \mathcal{S}_R(n-2), \mathrm{wt}(\boldsymbol{s}) \text{ mod } 2=0, \\
	& \sum_{i=1}^{\ceil{\frac{n}{2}}}w_i(\boldsymbol{s})=0 \text{ mod } 3, \text{ where } s_1^*\leq s^*_n\}.
	\end{split}\nonumber
\end{equation}
A similar construction exists for even $n$. The size of this code equals $\frac{|\mathcal{S}_R(n-2)|}{2}$. However, subsequently in Section \ref{sec::sdel} we conclude by means of Lemma \ref{lem::sym_del}, that the code $\mathcal{S}_R(n)$ is also capable of correcting a single composition error.

\emph{Construction 3 \cite{pattabiraman}:} A codebook $\mathcal{S}_{CA}^{(t)}(n)$ that is capable of rectifying {$t$-asymmetric substitution errors} is proposed in \cite{pattabiraman}, and for the sake of brevity, we henceforth call it a {$t$-asymmetric composition code}. $\mathcal{S}_{CA}^{(t)}(n)$ constitutes all codewords $\boldsymbol{s}=(\boldsymbol{\tilde{s}}_1^{ m/2} \boldsymbol{b}_1^{n-m} \boldsymbol{\tilde{s}}_{m/2+1}^{m})$, such that the components $\boldsymbol{\tilde{s}}_1^{m}$ and $\boldsymbol{b}_1^{n-m}$ are constructed as follows:
\begin{itemize}
	\item We choose $\boldsymbol{\boldsymbol{\tilde{s}}}=(\boldsymbol{\tilde{s}}^{m/2}_1 \boldsymbol{\tilde{s}}_{m/2+1}^m) \in \mathcal{S}^{(t)}_R(m)$, described by the sequence $\boldsymbol{\sigma}_{\boldsymbol{\tilde{s}}}$.
	\begin{equation}
	\begin{split}
	\mathcal{S}^{(t)}_R(m)=
	&\{\boldsymbol{s}\in \{0,1\}^m, \boldsymbol{s}^t_1=\boldsymbol{0}, \boldsymbol{s}^m_{m-t+1}=\boldsymbol{1}, \text{ and} \\
	& \exists I \subset \{t+1,\ldots, m-t\} \text{ such that}\\
	& \quad \quad \quad \quad \text{ for all } i \in I, s_i\neq s_{m+1-i}, \\
	& \quad \quad \quad \quad \text{ for all } i \notin I, s_i=s_{m+1-i},\\
	& \;\; \boldsymbol{s}_{[m/2]\cap I} \text{ is a Catalan-Bertrand string.}\}\\
	\end{split}\label{eq::srt}
	\end{equation}
	\item A systematic Reed-Solomon code over the alphabet $\{0,1,2\}$ is used to map $\boldsymbol{\sigma}_{\boldsymbol{\tilde{s}}}$ to a sequence $\boldsymbol{\sigma}_{s}$ by appending the values $(\sigma_{m/2+1}, \ldots, \sigma_{n/2})$, which help to construct $\boldsymbol{b}=\boldsymbol{b}_1^{n-m}$ as follows:
	\begin{equation}
		b_k b_{n-k+1} =\begin{cases}
		00, & \text{if } \sigma_{m/2+k}=0.\\
		01, & \text{if } \sigma_{m/2+k}=1.\\
		11, & \text{if } \sigma_{m/2+k}=2.
		\end{cases} \nonumber
	\end{equation}
	where $k \in [(n-m)/2]$.
\end{itemize}

The upcoming construction, designed to correct substitution errors in symmetric multiset pairs, exploits a bivariate generating polynomial representation $P_{\boldsymbol{s}}(x,y)$ of string $\boldsymbol{s}$, that works as follows. Let the first term always be $\big(P_{\boldsymbol{s}}(x,y)\big)_{0}=1$. Now by representing bits $0$ and $1$ as $y$ and $x$ respectively, we define the subsequent terms as:
\begin{equation}
	\big(P_{\boldsymbol{s}}(x,y)\big)_{i}=\begin{cases}
	y \big(P_{\boldsymbol{s}}(x,y)\big)_{i-1}, & \text{ if } s_i=0 \\
	x \big(P_{\boldsymbol{s}}(x,y)\big)_{i-1}, & \text{ if } s_i=1.
	\end{cases} \nonumber
\end{equation}
\begin{env_example}
	For $\boldsymbol{s}=001010111$, the bivariate generating polynomial is given by
	$P_{\boldsymbol{s}}(x,y)=1+y+y^2+xy^2+xy^3+x^2y^3+x^2y^4+x^3y^4+x^4y^4+x^5y^4$.
\end{env_example}

The corresponding construction can be defined more explicitly as follows. A code $\mathcal{S}_{CS}^{(t)}$ is called a \emph{$t$-symmetric composition code}, if for all $\boldsymbol{s}$, $\boldsymbol{v} \in \mathcal{S}_{CS}^{(t)}$, there exists no $\mathcal{I} \subseteq [\ceil{\frac{n}{2}}]$ with $|\mathcal{I}| \leq t$ such that 
\begin{align}
    \big| \bigcup_{i \in \mathcal{I}} (\widetilde{C}_i(\boldsymbol{s}) \setminus \widetilde{C}_i(\boldsymbol{v})\big|&\leq t, \nonumber \\
    \widetilde{C}_i(\boldsymbol{s})&= \widetilde{C}_i(\boldsymbol{v}) \quad \forall \; i \in \Big[\Big\lceil\frac{n}{2}\Big\rceil\Big]\setminus \mathcal{I}. \nonumber
\end{align}

\emph{Construction 4 \cite{pattabiraman}:} The authors of \cite{pattabiraman} also suggest a construction that corrects any $t$ symmetric composition substitutions in an entire composition multiset as follows.
\begin{equation}
\begin{split}
\mathcal{S}_{CS}^{(t)}(n)=
&\{\boldsymbol{s}\in \{0,1\}^n, \text{ s.t. } P_{\boldsymbol{s}}(\alpha^{\ell_1}, \alpha^{\ell_2})=a_{\ell_1,\ell_2},\\
& \quad \quad \mathrm{wt}(\boldsymbol{s})\equiv a \mod (2t+1)\}\\
\end{split}
\end{equation}
for all $\ell_1,\ell_2 \in \{0,1,\ldots, 4t\}$, $a \in \{0,1,\ldots,2t\}$ and where $(a_{\ell_1,\ell_2})_{\ell_1,\ell_2=0}^{4t}$ denotes a random vector from $\mathbb{F}_q^{(4t+1)^2}$. 

\section{New Error Models} \label{sec::error_models}
The subsequent sections explore error models that involve corrupting a valid composition multiset via the insertion or deletion of one or more multisets. 
\begin{defn}
	An \textbf{asymmetric multiset deletion} is said to have occurred in the composition multiset $C(\boldsymbol{s})$ of a string $\boldsymbol{s} \in \{0,1\}^n$, if for some $i \in [n]$, the multiset $C_i (\boldsymbol{s})$ is entirely missing, while $C_{n-i+1} (\boldsymbol{s})$ is uncorrupted. \label{def:amdel}
\end{defn}
\begin{defn}
	A \textbf{pair of symmetric multiset deletions} is said to have occurred if the composition multiset $C(\boldsymbol{s})$ of a string $\boldsymbol{s} \in \{0,1\}^n$, if for some $i \in [n]$ such that $i \neq n-i+1$, the multisets $C_i(\boldsymbol{s})$ and $C_{n-i+1}(\boldsymbol{s})$ are entirely eliminated.
\end{defn}
\begin{env_example}
	Let $\boldsymbol{s}=001010111$. If the composition multiset $C(\boldsymbol{s})$ is corrupted to
	\begin{align}
	    C'(\boldsymbol{s})=&\bigcup_{i \in [n] \setminus \{3\}} C_i(\boldsymbol{s}), \nonumber\\
	    =\{&0,0,1,0,1,0,1,1,1, 0^2,0^11^1,0^11^1,0^11^1,0^11^1,\nonumber \\
		&0^11^1,1^2,1^2,0^31^1,0^21^2,0^21^2,0^21^2,0^11^3,0^11^3,\nonumber\\
		&0^31^2,0^31^2,0^21^3,0^21^3,0^11^4,0^41^2,0^31^3,0^21^4,\nonumber\\
		&0^21^4,0^41^3,0^31^4,0^21^5,0^41^4,0^31^5,0^41^5\}.\nonumber
	\end{align}
	then an asymmetric multiset deletion is said to have occurred. More specifically, the multiset $C_3(\boldsymbol{s})=\{0^21^1,0^21^1,0^11^2,0^21^1,0^11^2,0^11^2,1^3\}$ has been deleted. On the other hand, if
	\begin{align}
	    C'(\boldsymbol{s})=&\bigcup_{i \in [n] \setminus \{3,7\}} C_i(\boldsymbol{s}), \nonumber\\
	    =\{&0,0,1,0,1,0,1,1,1, 0^2,0^11^1,0^11^1,0^11^1,0^11^1,\nonumber \\
		&0^11^1,1^2,1^2,0^31^1,0^21^2,0^21^2,0^21^2,0^11^3,0^11^3,\nonumber\\
		&0^31^2,0^31^2,0^21^3,0^21^3,0^11^4,0^41^2,0^31^3,0^21^4,\nonumber\\
		&0^21^4,0^41^4,0^31^5,0^41^5\}.\nonumber
	\end{align}
	we say that a pair of symmetric multiset deletions has occurred. Here compared to $C(\boldsymbol{s})$, we are missing the multisets $C_3(\boldsymbol{s})=\{0^21^1,0^21^1,0^11^2,0^21^1,0^11^2,0^11^2,1^3\}$ and $C_7(\boldsymbol{s})=\{0^41^3,0^31^4,0^21^5\}$.
\end{env_example}
\begin{defn}
	A composition multiset $C(\boldsymbol{s})$ of a string $\boldsymbol{s} \in \{0,1\}^n$ is said to have suffered a \textbf{composition insertion error}, if for some $i \in [n]$ the multiset $C_i (\boldsymbol{s})$ contains $n-i+2$ compositions, i.e. an unknown and invalid composition has been registered.
\end{defn}
\begin{env_example}
	Once again, let $\boldsymbol{s}=001010111$. If $C_7(\boldsymbol{s})$ has been altered as follows,
	\begin{align}
	    C'_7(\boldsymbol{s})&=\{0^41^3,0^31^4,0^21^5, 0^11^6\}. \nonumber
	\end{align}
	 we say that a composition insertion error has taken place.
\end{env_example}
The main contribution of this work consists of studying the aforementioned error models and proposing new coding constraints to combat the same. We also establish an equivalence between codes that correct composition insertions and composition deletions. Consequently, we restrict our attention to the latter for the remainder of this paper. 

To this end, we first propose the following composition reconstruction code that allows for the correction of $t$ asymmetric multiset deletions. Specifically, a code $\mathcal{S}_{DA}^{(t)}$ is termed as a \emph{$t$-asymmetric multiset deletion composition code}, if for all $\boldsymbol{s}$, $\boldsymbol{v} \in \mathcal{S}_{DA}^{(t)}$, there exists no $\mathcal{I} \subseteq [n]$ with $|\mathcal{I}| \leq t$ such that for all $i \in \mathcal{I}$,
\begin{align}
    C_i(\boldsymbol{s}) &\neq C_i(\boldsymbol{v}), \nonumber \\
    C_{n-i+1}(\boldsymbol{s}) &= C_{n-i+1}(\boldsymbol{v}), \nonumber\\
    C_{j}(\boldsymbol{s}) &= C_{j}(\boldsymbol{v})\quad \forall\; j \in [n]\setminus \mathcal{I}. \nonumber
\end{align}

\emph{Construction 5:} 
\begin{equation}
\begin{split}
{\mathcal{S}}^{(t)}_{DA}(n)=&\{\boldsymbol{s}\in \{0,1\}^n, s_1=0, s_n=1, \text{ and} \\
& \exists I \subset \{2,\ldots, \frac{n}{2}\},\; |I| \geq t, \text{ such that}\\
& \quad \quad \quad \quad \forall\; i \in I, s_i\neq s_{n+1-i}, \\
& \quad \quad \quad \quad \text{ and }\forall i \notin I, s_i=s_{n+1-i},\\
& \quad \quad \boldsymbol{s}_{[n/2]\cap I} \text{ is a string wherein each }\\ 
&\quad \text{prefix has at least $t$ more $0$s than $1$s.}\}
\end{split} \label{eq::srpt}
\end{equation}

The corresponding proof follows behind Theorem \ref{lem::t_dels}. Evidently, this construction is inspired from (\ref{eq::srt}), in that it requires at least $t$ $0$s in $\boldsymbol{s}_1^{n/2}$ and at least $t$ $1$s in $\boldsymbol{s}_{n/2+1}^n$, however their locations are not necessarily restricted as in (\ref{eq::srt}). The extension to odd codeword lengths is similar to (\ref{eq:sr_odd}).

Following this, we investigate the case of symmetric multiset deletions, and discover that when two or more symmetric multiset pairs are missing, additional constraints are needed to bolster the code $S_R(n)$ so as to guarantee unique reconstructability. In this context, a code $\mathcal{S}_{DS}^{(t)}$ is termed as a \emph{$t$-symmetric multiset deletion composition code}, if for all $\boldsymbol{s}$, $\boldsymbol{v} \in \mathcal{S}_{DS}^{(t)}$, there exists no $\mathcal{I} \subseteq \Big[\Big\lceil \frac{n}{2} \Big \rceil \Big]$ with $|\mathcal{I}| \leq t$ such that
\begin{align}
    \widetilde{C}_i(\boldsymbol{s}) &\neq \widetilde{C}_i(\boldsymbol{v}), \forall\; i \in \mathcal{I} \nonumber \\
    C_{i}(\boldsymbol{s}) &= C_{i}(\boldsymbol{v})\quad \forall\; i \in \Big[\Big\lceil \frac{n}{2} \Big \rceil \Big]\setminus \mathcal{I}. \nonumber
\end{align}

For the elementary case of two deleted symmetric multiset pairs, we propose the following code.

\emph{Construction 6:} 
\begin{equation}
\begin{split}
\mathcal{S}^{(2)}_{DS}(n)=&\{\boldsymbol{s}\in \mathcal{S}_R(n), \\
&  \sum_{i=1}^{\ceil{\frac{n}{2}}} w_i(\boldsymbol{s}) \text{ mod } 7=a,\; 0\leq a \leq 6\}.\\
\end{split} \label{eq::sds}
\end{equation}
Theorem \ref{lem::sds2a} proves that this code can indeed correct the deletion of two  symmetric multiset pairs.
We also generalize this construction to accommodate for the deletion of any $t$ consecutive symmetric multiset pairs, where $t \geq 2$. More explicitly, a code $\mathcal{S}_{DS}^{\prime (t)}$ is termed as a \emph{$t$-symmetric consecutive multiset deletion composition code}, if for all $\boldsymbol{s}$, $\boldsymbol{v} \in \mathcal{S}_{DS}^{\prime (t)}$, there exists no $\mathcal{I}=\{i, i+1, \ldots i+p-1\} \subseteq \Big[\Big\lceil \frac{n}{2} \Big \rceil \Big]$ with $p \leq t$ such that
\begin{align}
    \widetilde{C}_j(\boldsymbol{s}) &\neq \widetilde{C}_j(\boldsymbol{v}), \forall\; j \in \mathcal{I} \nonumber \\
    C_{j}(\boldsymbol{s}) &= C_{j}(\boldsymbol{v})\quad \forall\; j \in \Big[\Big\lceil \frac{n}{2} \Big \rceil \Big]\setminus \mathcal{I}. \nonumber
\end{align}

\emph{Construction 7:} 
\begin{equation}
\begin{split}
\mathcal{S}^{\prime (t)}_{DS}(n)=&\{\boldsymbol{s}\in \mathcal{S}_R(n),  \sum_{i=1}^{\frac{m}{2}} w_i(\boldsymbol{s}) \text{ mod } A=a,\\
& 0\leq a \leq A-1\}\\
\end{split} \label{eq::cdst}
\end{equation}
where $t\geq 2$ and

\begin{equation}
A=\Big\lceil \frac{4t^3}{3}+\frac{2t}{3}-\frac{31}{4} \Big\rceil. \nonumber
\end{equation}
 Theorem \ref{lem::sdst_final} proves that $\mathcal{S}^{\prime (t)}_{DS}(n)$ is capable of correcting the deletion of $t$ consecutive symmetric multiset pairs.

\begin{defn}
	A composition multiset $C(\boldsymbol{s})$ of the string $\boldsymbol{s} \in \{0,1,\}^n$ is said to have suffered an \textbf{asymmetric skewed substitution error}, if for some $i \in [n]$, a single composition of multiset $C_i(\boldsymbol{s})$ is replaced with one of a lower Hamming weight, such that the symmetric counterpart $C_{n-i+1}(\boldsymbol{s})$ remains unaffected.
\end{defn}
\begin{env_example}
    For instance, if an erroneous measurement corrupts the composition $0^21^4$, the measured compositions could be $0^31^3$ or $0^41^2$, but not $0^11^5$.
\end{env_example}
Formally, a code $\mathcal{C}^{\prime(t)}$ is referred to as a \emph{$t$-asymmetric skewed composition code}, if for all $\boldsymbol{s}$, $\boldsymbol{v} \in \mathcal{C}^{\prime(t)}$, there exists no $\mathcal{I} \subseteq [n]$ with $|\mathcal{I}| \leq t$ such that for all $i \in \mathcal{I}$,
\begin{align}
    C_i(\boldsymbol{s}) &\neq C_i(\boldsymbol{v}), \nonumber \\
    C_{n-i+1}(\boldsymbol{s}) &= C_{n-i+1}(\boldsymbol{v}), \nonumber\\
    C_{j}(\boldsymbol{s}) &= C_{j}(\boldsymbol{v})\quad \forall\; j \in [n]\setminus \mathcal{I} \nonumber
\end{align}

We subsequently prove in Lemma~\ref{lem::skewc} of Section~\ref{sec::skew} that the code $\mathcal{S}^{(t)}_{DA}(n)$ (Construction 5) is sufficiently robust to allow the correction of $t$ skewed asymmetric substitution errors in its composition set.

These results, along with the earlier constructions proposed in~\cite{pattabiraman, p2, p3}, have been summarized in Table~\ref{tab::results}.

\begin{table*}[!htb]
	\begin{tabular}{|c | c | c | c | c || } 
		\hline
		Code & Symbol &  \makecell{Upper bound \\ on redundancy} & Proof  \\  
		\hline\hline
		\makecell{Composition-\\reconstructable code}& $\mathcal{S}_{R}(n)$ & $\frac{1}{2}\log_2 n +5$ & \cite{pattabiraman,p2} \\
		\hline
		\makecell{Single composition \\error-correcting code}& $\mathcal{S}^{(1)}_{CA}(n)$ & $\frac{1}{2}\log_2 (n-2) +8$ & \cite{pattabiraman,p2} \\
		\hline
		\makecell{$t$-asymmetric \\ composition code} &$\mathcal{S}^{(t)}_{CA}(n)$ & $\Big(\frac{1}{2}+3t \Big) \log_2 n +2t+5$ & \cite{pattabiraman} \\ 
		\hline
		\makecell{$t$-composition code} & $\mathcal{S}_{CS}^{(t)}(n)$ & $156t^2 \log_2 n$ & \cite{pattabiraman, p3} \\ 
		\hline
		\makecell{$t$-asymmetric multiset \\deletion composition code} &$\mathcal{S}^{(t)}_{DA}(n)$ &$\frac{1}{2} \log_2{(n-2t)}+2t+3 $ & Th. \ref{lem::t_dels} \\ 
		\hline
		 \makecell{$2$-symmetric multiset \\deletion composition code}  & $\mathcal{S}^{(2)}_{DS}(n)$ & $\frac{1}{2} \log_2{(n-2)}+8 $ & Th. \ref{lem::sds2a} \\
		\hline
		\makecell{$t$-symmetric consecutive \\multiset deletion composition code} & $\mathcal{S}^{\prime (t)}_{DS}(n)$ & \makecell{$\frac{1}{2} \log_2(n-2)$ \\[1ex] $+\log_2 \Big\lceil \frac{4t^3}{3}+\frac{2t}{3}-\frac{31}{4} \Big\rceil +5$} &  Th. \ref{lem::sdst_final} \\
		\hline 
	\end{tabular}
	\caption{Summary of constructions}
	\label{tab::results}
\end{table*}

\section{Code Equivalence: Insertion and Deletion of Multisets} \label{sec::equiv}
In this section, we demonstrate how codes which can correct the deletion of a group of $t$ multisets, can also correct the occurrence of insertion errors in those $t$ multisets.
\begin{lemma}
	A code can correct the deletion of $t$ composition multisets, if and only if it can correct any number of composition insertion errors in those $t$ multisets.
\end{lemma}
\begin{proof}
	We prove this by contradiction. Let there be two binary strings $\boldsymbol{s}, \boldsymbol{v} \in \mathcal{S}_R(n)$, such that:
	\begin{equation}
	D_t(\boldsymbol{s}) \cap D_t(\boldsymbol{v}) \neq \emptyset. \label{eq::dtsv}
	\end{equation}
	where $D_t(\boldsymbol{s})$ constitutes all codewords in $\mathcal{S}_R(n)$ that $\boldsymbol{s}$ becomes equicomposable with upon the deletion of at most $t$ multisets, i.e.,
	\begin{equation}
	\begin{split}
	    D_t(\boldsymbol{s}) =&\{\boldsymbol{u} \in \mathcal{S}_R(n) \text{ such that } \exists \;  \mathcal{I} \subseteq [n],\;  |\mathcal{I}| \leq t, \\
	    & \bigcup_{ i \in [n] \setminus \mathcal{I}} C_i(\boldsymbol{s})=\bigcup_{ i \in [n] \setminus \mathcal{I}}C_i(\boldsymbol{u}) \}.
	\end{split} \nonumber
	\end{equation}
	Equation \eqref{eq::dtsv}
	implies that at least $n-t$ composition multisets of $\boldsymbol{s}$ and $\boldsymbol{v}$ are identical. In other words, when a specific group of $t$ multisets disappears from the multiset information of $\boldsymbol{s}$ and $\boldsymbol{v}$, they become indistinguishable. Let these differing multisets correspond to substring lengths $i_1, i_2, \ldots i_t$. This allows us to write that:
	\begin{equation}
	\bigcup_{j\in [n] \backslash \{i_1, \ldots i_t\}} C_j(\boldsymbol{s})=\bigcup_{j\in [n] \backslash \{i_1, \ldots i_t\}} C_j(\boldsymbol{v}). \nonumber
	\end{equation}
	If we perform a set union operation on both sides of the previous equation with $\bigcup_{i \in \{i_1, \ldots i_t\}}C_i(\boldsymbol{s}) \cup C_i(\boldsymbol{v})$, then we get:
	\begin{eqnarray}
	&&\bigcup_{i \in \{i_1, \ldots i_t\}} (C_j(\boldsymbol{v}) \backslash C_j(\boldsymbol{s})) \cup \bigcup_{j\in [n]}  C_j(\boldsymbol{s}) \nonumber \\
	&=&\bigcup_{i \in \{i_1, \ldots i_t\}} (C_j(\boldsymbol{s})\backslash C_j(\boldsymbol{v})) \cup \bigcup_{j\in [n] } C_j(\boldsymbol{v}). \nonumber
	\end{eqnarray}
	This effectively means that if the multisets $C_{i_1}(\boldsymbol{s}), \ldots C_{i_t}(\boldsymbol{s})$ are corrupted by the insertion of some specific erroneous compositions, then the multiset information may correspond to both $\boldsymbol{s}$ and $\boldsymbol{v}$, and vice-versa.
	This lets us write the following:
	\begin{equation}
	I_t(\boldsymbol{s}) \cap I_t(\boldsymbol{v}) \neq \emptyset.
	\end{equation}
	where $I_t(\boldsymbol{s})$ denotes the set of all codewords $\boldsymbol{u} \in \mathcal{S}_R(n)$ whose composition multisets, upon suffering any number of insertion errors in at most $t$ distinct multisets, resemble $C(\boldsymbol{s})$ after corruption by certain composition insertions in those affected multisets. In other words, at least $n-t$ distinct multisets of $\boldsymbol{s}$ and $\boldsymbol{u}$ are identical. Consequently, we can write
	\begin{equation}
	\begin{split}
	    I_t(\boldsymbol{s}) =&D_t(\boldsymbol{s}) \\
	    =&\{\boldsymbol{u} \in \mathcal{S}_R(n) \text{ such that } \exists \;  \mathcal{I} \subseteq [n],\;  |\mathcal{I}| \leq t, \\
	    & \forall \; i \in [n] \setminus \mathcal{I},\;\; C_i(\boldsymbol{s})=C_i(\boldsymbol{u}) \}
	\end{split} \nonumber
	\end{equation}
\end{proof}
Owing to this result, we deem it sufficient to focus on multiset deletion-correcting codes.
The subsequent sections examine how multiset deletions affect the reconstructability of an encoded string drawn from $\mathcal{S}_R(n)$. Similar to \cite{pattabiraman}, we categorize such deletion errors into two major settings. 
\section{Asymmetric Multiset Deletion-correcting Composition-Reconstruction Codes} \label{sec::adel}




We begin by considering an error model where a complete multiset $C_k(\boldsymbol{s})$ can be deleted from the composition multiset $C(\boldsymbol{s})$. This is formally referred to as a single asymmetric multiset deletion [see Definition \ref{def:amdel}]. We investigate whether the reconstruction codebook [see Construction 1] guarantees unique recoverability under this model.
To proceed in this direction, we first take note of the following lemma, which results from a specific case of {\cite[Lemma 4]{pattabiraman}}.
\begin{lemma}
	Let $\boldsymbol{s},\boldsymbol{v}\in \mathcal{S}_R(m)$ share the same $\boldsymbol{\sigma}$ sequence and satisfy $|C_j(\boldsymbol{s})\backslash C_j(\boldsymbol{v})| \leq 2$ for all $j \in [m]$. If the longest prefix-suffix pair shared by $\boldsymbol{s}$ and $\boldsymbol{v}$ is of length $i$, then their corresponding composition multisets $C_{m-i-1}$ and $C_{m-i-2}$ each differ in at least 2 compositions. \label{lem::comp_mismatch}
\end{lemma}
To shortly highlight the implications of this lemma, we consider the strings $\boldsymbol{s}=001011101$ and $\boldsymbol{v}=001110101$. Clearly, they are both specified by $\boldsymbol{\sigma}=(1,0,2,1,1)$. Since the longest prefix-suffix pair shared by them is $(001,101)$, i.e., of length $3$, their respective multisets $C_{4}$ and $C_5$ differ by at least 2 compositions.
\begin{lemma}
	Consider a string $\boldsymbol{s} \in \mathcal{S}_{R}(n)$. Given $C'(\boldsymbol{s})=\bigcup_{i\in [n] \backslash \{k\}} C_i(\boldsymbol{s})$ for any $k \in [n]$, $\boldsymbol{s}$ can be fully recovered. \label{lem::del_set1}
\end{lemma}
\begin{proof}
	\begin{mycases}
		\case
		\emph{$n$ is even}
		
		From the steps of the reconstruction algorithm as described in Section \ref{subsec::rec}, it is evident that we only require the composition multisets $C_n(\boldsymbol{s}),\ldots, C_{\frac{n}{2}}(\boldsymbol{s})$. Hence, if $k<\frac{n}{2}$, the reconstruction of $\boldsymbol{s}$ is straightforward. On the contrary, if $k\geq \frac{n}{2}$, one can still infer the cumulative weight of the missing multiset $C_k(\boldsymbol{s})$ from (\ref{eq::wt_sym}). 
		Consequently, $\boldsymbol{\sigma}_{s}$ can be obtained accurately. 
		
		In the absence of $C_k(\boldsymbol{s})$, the prefix and suffix can be constructed upto $\boldsymbol{s}_1^{n-k-1}$ and $\boldsymbol{s}_{k+2}^{n}$. When $\sigma_{k}=\sigma_{n-k+1} \in \{0,2\}$, there remains no ambiguity concerning the bits $s_{n-k}$ and $s_{k+1}$. However, when $\sigma_{k}=1$, one can either have $(s_{n-k}, s_{k+1})=(0,1)$ or $(s_{n-k}, s_{k+1})=(1,0)$ if both of these possibilities guarantee weight mismatch between $\boldsymbol{s}_1^{n-k}$ and $\boldsymbol{s}_{k+1}^{n}$. Now since $\boldsymbol{s} \in \mathcal{S}_R(n)$, Lemma \ref{lem::comp_mismatch} tells us that choosing the bits $s_{n-k}$ and $s_{k+1}$ incorrectly, will lead to an incompatibility with the multiset $C_{k-1}(\boldsymbol{s})$. Thus there exists only one valid choice for these bits, implying that $\boldsymbol{s}$ is uniquely recoverable.		
		\case\emph{$n$ is odd}
		
		Similar to the previous case, it can be argued that for any missing composition multiset $C_k(\boldsymbol{s})$, where $k \neq \ceil{\frac{n}{2}}$, $\boldsymbol{s}$ can be easily and uniquely determined. The more interesting case occurs when $k=\ceil{\frac{n}{2}}$, since the absence of $C_{\ceil{\frac{n}{2}}}(\boldsymbol{s})$, and thus $w_{\ceil{\frac{n}{2}}}(\boldsymbol{s})$, prevents us from computing $\sigma_{\ceil{\frac{n}{2}}-1}$ and $\sigma_{\ceil{\frac{n}{2}}}$. However, their sum is known from (\ref{eq::w1}), i.e. 
		\begin{equation}
		\sigma_{\ceil{\frac{n}{2}}-1}+\sigma_{\ceil{\frac{n}{2}}}=w_1(\boldsymbol{s})-\sum_{i=1}^{\ceil{\frac{n}{2}}-2} \sigma_i.
		\end{equation}
		Since $\sigma_{\ceil{\frac{n}{2}}-1}=\mathrm{wt}(s_{\ceil{\frac{n}{2}}-1}s_{\ceil{\frac{n}{2}}+1}) \in \{0,1,2\}$ and $\sigma_{\ceil{\frac{n}{2}}}=\mathrm{wt}(s_{\ceil{\frac{n}{2}}}) \in \{0,1\}$, these values can be inferred directly when $\sigma_{\ceil{\frac{n}{2}}-1}+\sigma_{\ceil{\frac{n}{2}}} \in \{0,3\}$. However, an ambiguity arises when $\sigma_{\ceil{\frac{n}{2}}-1}+\sigma_{\ceil{\frac{n}{2}}} \in \{1,2\}$.
		
		Let $\boldsymbol{v} \in \mathcal{S}_R(n)$ be a string with which $\boldsymbol{s}$ becomes equicomposable when the multiset $C_{\ceil{n/2}}$ is deleted, i.e.,
		\begin{equation}
		\bigcup_{i\in [n] \backslash \{\ceil{\frac{n}{2}}\}} C_i(\boldsymbol{s})=\bigcup_{i\in [n] \backslash \{\ceil{\frac{n}{2}}\}} C_i(\boldsymbol{v}). \label{eq::svb}
		\end{equation}
		Also, let $\boldsymbol{v}$ be specified by $\boldsymbol{\sigma}_{\boldsymbol{v}}=(\sigma'_1, \ldots, \sigma'_{\ceil{n/2}})$. As a consequence of (\ref{eq::svb}), we can write:
		\begin{equation}
			\begin{split}
				\sigma_i&=\sigma'_i, \quad \forall \;\; 1\leq i \leq \ceil{\frac{n}{2}}-2  \\
				\sigma_{\ceil{\frac{n}{2}}-1}+\sigma_{\ceil{\frac{n}{2}}}&=\sigma'_{\ceil{\frac{n}{2}}-1}+\sigma'_{\ceil{\frac{n}{2}}}. 
			\end{split} \label{eq::svb2}
		\end{equation} 
		To verify whether the reconstructability of $\boldsymbol{s}$ is affected, we simply check if there exists a suitable $\boldsymbol{v}$ that satisfies (\ref{eq::svb}) and (\ref{eq::svb2}). We also note that (\ref{eq::svb}) directly implies the equality of the prefix-suffix pairs $(\boldsymbol{s}_1^{\ceil{\frac{n}{2}}-2}, \boldsymbol{s}_{\ceil{\frac{n}{2}}+2}^n)=(\boldsymbol{v}_1^{\ceil{\frac{n}{2}}-2}, \boldsymbol{v}_{\ceil{\frac{n}{2}}+2}^n)$.

		\begin{figure}[!htb]
			\centering			
			\scalebox{0.8}{\begin{tikzpicture}[font=\ttfamily,
				array/.style={matrix of nodes,nodes={draw, minimum size=7mm, fill=green!30},column sep=-\pgflinewidth, row sep=0.5mm, nodes in empty cells,
					row 1/.style={nodes={draw=none, fill=none, minimum size=5mm}},
					row 1 column 1/.style={nodes={draw}}}]
				\draw (0,0) rectangle (2,0.6) node[pos=.5] {$\boldsymbol{s}_1^{\ceil{\frac{n}{2}}-2}$};
				
				\draw (2,0) rectangle (3,0.6) node[pos=0.5] {$1-b$};
				\draw (3,0) rectangle (4,0.6) node[pos=0.5] {$b$};
				\draw (4,0) rectangle (5,0.6) node[pos=0.5] {$1-b$};					
				
				\draw (5,0) rectangle (7,0.6) node[pos=0.5] {$\boldsymbol{s}_{\ceil{\frac{n}{2}}+2}^n$};
				
				\draw (0,-0.4) rectangle (2,-1) node[pos=0.5] {$\boldsymbol{v}_1^{\ceil{\frac{n}{2}}-2}$};
				
				\draw (2,-0.4) rectangle (3,-1) node[pos=0.6] {$v_{+}$};
				\draw (3,-0.4) rectangle (4,-1) node[pos=0.5] {$1-b$};
				\draw (4,-0.4) rectangle (5,-1) node[pos=0.6] {$v_{-}$};
				\draw (5,-0.4) rectangle (7,-1) node[pos=0.5] {$\boldsymbol{v}_{\ceil{\frac{n}{2}}+2}^n$};
				\end{tikzpicture}}
			\caption{Strings $\boldsymbol{s}$ and $\boldsymbol{v}$ are such that $(\boldsymbol{s}_1^{\ceil{\frac{n}{2}}-2}, \boldsymbol{s}_{\ceil{\frac{n}{2}}+2}^n)=(\boldsymbol{v}_1^{\ceil{\frac{n}{2}}-2}, \boldsymbol{v}_{\ceil{\frac{n}{2}}+2}^n)$, where $v_{+}=1-v_{-}$.}
			\label{fig::f1stog}
		\end{figure}

		We jointly depict the specific subcases in Fig. \ref{fig::f1stog}, wherein we allow for $\sigma_{\ceil{\frac{n}{2}}-1}+\sigma_{\ceil{\frac{n}{2}}} \in \{1,2\}$
		since for both $\boldsymbol{s}$ and  $\boldsymbol{v}$, we have: 
		\begin{equation}
		\sigma_{\ceil{\frac{n}{2}}-1}+\sigma_{\ceil{\frac{n}{2}}}=2-b. \nonumber
		\end{equation}
		where $b \in \mathbb{F}_2$. To proceed with the proof, we try to determine the conditions under which $C_{\ceil{\frac{n}{2}}-1}(\boldsymbol{s})=C_{\ceil{\frac{n}{2}}-1}(\boldsymbol{v})$ holds. This would require the following set equality:
		\begin{equation}
		\left\{
		\begin{split}
		&\{c(\boldsymbol{s}_1^{\ceil{\frac{n}{2}}-2}), 1-b\}\\
		&\{c(\boldsymbol{s}_2^{\ceil{\frac{n}{2}}-2}), b,1-b\}\\
		&\{c(\boldsymbol{s}_{\ceil{\frac{n}{2}}+2}^n), 1-b\}\\
		&\{c(\boldsymbol{s}_{\ceil{\frac{n}{2}}+2}^{n-1}), b,1-b\}
		\end{split}\right\}=
		\left\{
		\begin{split}
		&\{c(\boldsymbol{v}_1^{\ceil{\frac{n}{2}}-2}), v_+\}\\
		&\{c(\boldsymbol{v}_2^{\ceil{\frac{n}{2}}-2}),v_+, 1-b\}\\
		&\{c(\boldsymbol{v}_{\ceil{\frac{n}{2}}+2}^n), 1-v_+\}\\
		&\{c(\boldsymbol{v}_{\ceil{\frac{n}{2}}+2}^{n-1}), 1-v_+,1-b\}.
		\end{split}\right\}. \nonumber
		\end{equation}
		By checking the above relation exhaustively for all possibilities of $(b, v_+) \in \{0,1\}^2$, we conclude that the multisets $C_{\ceil{n/2}-1}(\boldsymbol{s})$ and $C_{\ceil{n/2}-1}(\boldsymbol{s})$ can never match. Therefore, $\boldsymbol{v}$ does not exist and $\boldsymbol{s}$ retains its unique reconstructability.
	\end{mycases}
\end{proof}
It follows directly from the preceding lemma that
\begin{lemma}
	The code $\mathcal{S}_R(n)$ is a single asymmetric multiset deletion composition code. \label{lem::del_set1_new}
\end{lemma}

As a second step, $\mathcal{S}_R(n)$ is now generalized to $\mathcal{S}_{DA}^{(t)}(n)$ [see Construction 5] to allow correcting the deletion of $t$ asymmetric multisets. To prove why this construction works, we first consider the following lemma.
\begin{lemma}
	Let $\boldsymbol{s}, \boldsymbol{v} \in \mathcal{S}^{(t)}_{DA}(n)$ be specified by an identical $\boldsymbol{\sigma}$ sequence, such that the longest prefix-suffix pair shared by them is of length $k$. Then their corresponding multisets $C_{n-i-1}, \ldots, C_{n-i-t-1}$ differ by at least two compositions. \label{lem::t_dels_sets}
\end{lemma}
\begin{proof}
	Since $\boldsymbol{s}$ and $\boldsymbol{v}$ bear the same $\boldsymbol{\sigma}$ sequence and their prefix-suffix pair of length $k+1$ do not match, we conclude that $\sigma_{k+1}=1$ and $s_{k+1} \neq v_{k+1}$. Without loss of generality, we assume $s_{k+1}=0$ and it becomes obvious that $|C_{n-k-1}(\boldsymbol{s}) \backslash C_{n-k-1}(\boldsymbol{v})|=2$.

	\begin{figure}[!htb]
		\centering			
		\scalebox{0.8}{\begin{tikzpicture}[font=\ttfamily,
			array/.style={matrix of nodes,nodes={draw, minimum size=7mm, fill=green!30},column sep=-\pgflinewidth, row sep=0.5mm, nodes in empty cells,
				row 1/.style={nodes={draw=none, fill=none, minimum size=5mm}},
				row 1 column 1/.style={nodes={draw}}}]
			\draw (0,0) rectangle (2,0.6) node[pos=.5] {$\boldsymbol{s}_1^{k}$};
			\draw (2,0) rectangle (2.5,0.6) node[pos=0.5] {$0$};
			\draw (2.5,0) rectangle (3,0.6) node[pos=0.5] {$s_+$};
			\draw (3,0) rectangle (4.5,0.6) node[pos=0.5] {$\boldsymbol{s}_{k+3}^{n-k-2}$};
			\draw (4.5,0) rectangle (5,0.6) node[pos=0.5] {$s_-$};
			\draw (5,0) rectangle (5.5,0.6) node[pos=0.5] {$1$};
			\draw (5.5,0) rectangle (7.5,0.6) node[pos=0.5] {$\boldsymbol{s}_{n-k+1}^n$};
			
			\draw (0,-0.2) rectangle (2,-0.8) node[pos=0.5] {$\boldsymbol{v}_1^{k}$};
			\draw (2,-0.2) rectangle (2.5,-0.8) node[pos=0.6] {$1$};
			\draw (2.5,-0.2) rectangle (3,-0.8) node[pos=0.6] {$v_+$};
			\draw (3,-0.2) rectangle (4.5,-0.8) node[pos=0.6] {$\boldsymbol{v}_{k+3}^{n-k-2}$};
			\draw (4.5,-0.2) rectangle (5,-0.8) node[pos=0.6] {$v_-$};
			\draw (5,-0.2) rectangle (5.5,-0.8) node[pos=0.6] {$0$};
			\draw (5.5,-0.2) rectangle (7.5,-0.8) node[pos=0.5] {$\boldsymbol{v}_{n-k+1}^n$};
			\end{tikzpicture}}
		\caption{Strings $\boldsymbol{s}$ and $\boldsymbol{v}$ are related such that $(\boldsymbol{s}_1^{k}, \boldsymbol{s}_{n-k+1}^n)=(\boldsymbol{v}_1^{k}, \boldsymbol{v}_{n-k+1}^n)$ and $c(\boldsymbol{s}_{k+2}^{n-k-1})=c(\boldsymbol{v}_{k+2}^{n-k-1})$}
		\label{fig::mult_setst}
	\end{figure}

	As for the remaining multisets, we undertake the approach used in \cite[Lemma 4]{pattabiraman}, i.e., we design a set of strings $\mathcal{V}_{\boldsymbol{s}}$, such that for each $\boldsymbol{v} \in \mathcal{V}_{\boldsymbol{s}}$, $\boldsymbol{s}$ and $\boldsymbol{v}$ are specified by the same $\boldsymbol{\sigma}$ sequence, and satisfy:
	\begin{eqnarray}
	(\boldsymbol{s}_1^k, \boldsymbol{s}_{n-k+1}^n)&=&(\boldsymbol{v}_1^k, \boldsymbol{v}_{n-k+1}^n), \nonumber \\
	c(\boldsymbol{s}_{t+2}^{n-t-1})&=& c(\boldsymbol{v}_{t+2}^{n-t-1}), \label{eq::csv_eq} \\
	|C_{n-k-j}(\boldsymbol{s}) \backslash C_{n-k-j}(\boldsymbol{v})| &\leq& 2, \quad \quad \forall \;j \in [t+1]. \nonumber
	\end{eqnarray}
	Equation (\ref{eq::csv_eq}) follows directly from the premise of a common $\boldsymbol{\sigma}$ sequence.
	Similar to \cite[Lemma 4]{pattabiraman}, we note that $|C_{n-k-2}(\boldsymbol{s}) \backslash C_{n-k-2}(\boldsymbol{v})|$ is minimized when $\sigma_{k+2}=1$ and $(s_+,v_+)=(1,0)$,  thereby leading to $|C_{n-k-2}(\boldsymbol{s}) \backslash C_{n-k-2}(\boldsymbol{v})|=2$. Now, if an additional condition is upheld:
	\begin{eqnarray}
	(\boldsymbol{s}_{k+3}^{t+1}, \boldsymbol{s}_{n-t}^{n-k-2})&=&(\boldsymbol{v}_{k+3}^{t+1}, \boldsymbol{v}_{n-t}^{n-k-2}). \label{eq::csv_eq2}
	\end{eqnarray}
	we can show that $|C_{n-k-j}(\boldsymbol{s}) \backslash C_{n-k-j}(\boldsymbol{v})|=2$ for any $j \in [t+1]$, by examining the following set equality:
	\begin{eqnarray}
	&&\left\{\begin{split}
	&\{c(\boldsymbol{s}_1^{k}),01, c(\boldsymbol{s}_{k+3}^{n-k-j})\}\\
	&\{c(\boldsymbol{s}_2^{k}),01,  c(\boldsymbol{s}_{k+3}^{n-k-j+1})\}\\
	&\qquad \qquad \qquad \quad\vdots\\
	&\{c(\boldsymbol{s}_{j-1}^{k}),01,c(\boldsymbol{s}_{k+3}^{n-k-2})\}\\
	&\{c(\boldsymbol{s}_{j}^{k}),0^21,c(\boldsymbol{s}_{k+3}^{n-k-2})\}\\
	&\{c(\boldsymbol{s}_{n-k+1}^{n}),01, c(\boldsymbol{s}_{k+j+1}^{n-k-2})\}\\
	&\{c(\boldsymbol{s}_{n-k+1}^{n-1}),01, c(\boldsymbol{s}_{k+j}^{n-k-2})\}\\
	&\qquad \qquad \qquad \quad\vdots\\
	&\{c(\boldsymbol{s}_{n-k+1}^{n-j+2}),01, c(\boldsymbol{v}_{k+3}^{n-k-2})\}\\
	&\{c(\boldsymbol{s}_{n-k+1}^{n-j+1}),01^2, c(\boldsymbol{v}_{k+3}^{n-k-2})\}
	\end{split} \right\} \qquad \qquad\nonumber\\
	&&\qquad\qquad= \left\{\begin{split}
	&\{c(\boldsymbol{v}_1^{k}),01, c(\boldsymbol{v}_{k+3}^{n-k-j})\}\\
	&\{c(\boldsymbol{v}_2^{k}),01,  c(\boldsymbol{v}_{k+3}^{n-k-j+1})\}\\
	&\qquad \qquad \qquad \quad\vdots\\
	&\{c(\boldsymbol{v}_{j-1}^{k}),01,c(\boldsymbol{v}_{k+3}^{n-k-2})\}\\
	&\{c(\boldsymbol{v}_{j}^{k}),01^2,c(\boldsymbol{v}_{k+3}^{n-k-2})\}\\
	&\{c(\boldsymbol{v}_{n-k+1}^{n}),01, c(\boldsymbol{v}_{k+j+1}^{n-k-2})\}\\
	&\{c(\boldsymbol{v}_{n-k+1}^{n-1}),01, c(\boldsymbol{v}_{k+j}^{n-k-2})\}\\
	&\qquad \qquad \qquad \quad\vdots\\
	&\{c(\boldsymbol{v}_{n-k+1}^{n-j+2}),01, c(\boldsymbol{v}_{k+3}^{n-k-2})\}\\
	&\{c(\boldsymbol{v}_{n-k+1}^{n-j+1}),0^21, c(\boldsymbol{v}_{k+3}^{n-k-2})\}
	\end{split} \right\}. \label{eq::set_eq1a}
	\end{eqnarray}
	By exploiting (\ref{eq::csv_eq}) and (\ref{eq::csv_eq2}), one can simplify this to:
	\begin{eqnarray}
	&&\left\{\begin{split}
	&\{c(\boldsymbol{s}_{j}^{k}),0^21,c(\boldsymbol{s}_{k+3}^{n-k-2})\}\\
	&\{c(\boldsymbol{s}_{n-k+1}^{n-j+1}),01^2, c(\boldsymbol{s}_{k+3}^{n-k-2})\}
	\end{split} \right\} \qquad \qquad\nonumber\\
	&&\qquad\qquad= \left\{\begin{split}
	&\{c(\boldsymbol{s}_{j}^{k}),01^2,c(\boldsymbol{s}_{k+3}^{n-k-2})\}\\
	&\{c(\boldsymbol{s}_{n-k+1}^{n-j+1}),0^21, c(\boldsymbol{s}_{k+3}^{n-k-2})\}
	\end{split} \right\}. \nonumber
	\end{eqnarray}
	Upon combining (\ref{eq::csv_eq}) and (\ref{eq::csv_eq2}), further reduction is possible: 
	\begin{eqnarray}
	\left\{\begin{split}
	&\{c(\boldsymbol{s}_{j}^{k}),0^21\}\\
	&\{c(\boldsymbol{s}_{n-k+1}^{n-j+1}),01^2\}
	\end{split} \right\} &=& \left\{\begin{split}
	&\{c(\boldsymbol{s}_{j}^{k}),01^2\}\\
	&\{c(\boldsymbol{s}_{n-k+1}^{n-j+1}),0^21\}
	\end{split} \right\}. \label{eq::csvps}
	\end{eqnarray}
	We note that the preceding equality only holds if: $$c(\boldsymbol{s}_j^k)=c(\boldsymbol{s}_{n-k+1}^{n-j+1})$$
	However from Fig. \ref{fig::mult_setst} and the definition of $\mathcal{S}^{*(t)}_R(n)$ in (\ref{eq::srpt}), we observe that:
	\begin{eqnarray}
	\mathrm{wt}(\boldsymbol{v}_2^{k+1})+t \leq \mathrm{wt}(\boldsymbol{v}_{n-k+1}^{n-1}) \nonumber\\
	\implies \mathrm{wt}(\boldsymbol{v}_1^{k})+t+2 \leq \mathrm{wt}(\boldsymbol{v}_{n-k+1}^{n}). \nonumber
	\end{eqnarray}
	This inequality allows us to conclude that (\ref{eq::csvps}) never holds for any $j \in [t+1]$, consequently proving the statement of this lemma.
\end{proof}
The preceding lemma now helps us establish that the code $\mathcal{S}^{(t)}_{DA}(n)$ is robust to the deletion of any $t$ asymmetric multisets. 
\begin{theorem}
	Given the composition multisets $C_i(\boldsymbol{s})$ for $i~\in~[n]\backslash \{i_1, \ldots i_t\}$, where $\boldsymbol{s} \in \mathcal{S}^{(t)}_{DA}(n)$ [see Construction~5], such that no two of the deleted multisets are mutually symmetric, $\boldsymbol{s}$ can be uniquely recovered. \label{lem::t_dels}
\end{theorem}
\begin{proof}
	\begin{mycases}
		\case The deleted multisets are consecutive.
		This case is directly implied by Lemma \ref{lem::t_dels_sets}.
		\case All of the deleted multisets are not consecutive.\\
		Since the reconstruction algorithm functions in an outside-in manner, the missing multiset encountered first, corresponds to that of highest substring length. In the following analysis, we assume that $i_t >i_{t-1}> \ldots>i_1$.
		
		If $i_t=n$, we can directly infer $C_{n}(\boldsymbol{s})$ from the cumulative weight of $C_1(\boldsymbol{s})$. Alternatively when $i_t<n$ and additionally $i_t, \ldots, i_{t-j+1}$ are consecutive, the prefix-suffix pair $(\boldsymbol{s}_1^{n-i_t-1}, \boldsymbol{s}^n_{i_t+2})$ an incorrect assignment of the bit pair $(s_{n-i_t}, s_{i_t+1})$ will certainly cause an incompatibility with the multiset $C_{i_{t-j+1}-1}(\boldsymbol{s})=C_{i_{t}-j}(\boldsymbol{s})$, as Lemma \ref{lem::t_dels_sets} suggests. Thus, the backtracking algorithm can detect the mistake and accurately reconstruct the string upto $(\boldsymbol{s}_1^{n-i_t+j}, \boldsymbol{s}^n_{i_t-j+1})$. Absence of the other missing multisets $C_{i_{t-j}}, \ldots, C_{i_{1}}$ can be dealt with similarly.
	\end{mycases}
\end{proof}
The previous theorem implies the following.
\begin{theorem}
    $\mathcal{S}^{(t)}_{DA}(n)$ is a $t$-asymmetric multiset deletion composition code. \label{lem::t_dels_new}
\end{theorem}
We also bound the number of redundant bits required by $\mathcal{S}_{DA}^{(t)}$ as follows.
\begin{lemma}
	The code $\mathcal{S}_{DA}^{(t)}$ requires at most $\frac{1}{2} \log(n-2t)+2t+3$ bits of redundancy.
\end{lemma}
\begin{proof}
	We refer to (\ref{eq::srpt}) and additionally recount from \cite{pattabiraman} that $\frac{1}{2} \binom{2h}{h}$ indicates the number of all strings of length $2h$ wherein every prefix of which contains strictly more $0$s than $1$s. For odd lengths $2h+1$, this term serves as a lower bound. Similarly, to count all strings $\boldsymbol{s} \in \{0,1\}^p$ wherein each prefix (of length exceeding $t$) contains at least $t$ more $0$s than $1$s, we simply note that such strings satisfy  $\boldsymbol{s}_1^{t-1}=\boldsymbol{0}$ and $\boldsymbol{s}_{t}^{p}$ should be a standard Catalan-Bertrand string. By virtue of this, we derive a lower bound on dimension of the codebook:
	\begin{align}
	|\mathcal{S}^{(t)}_{DA}(n)|&\geq \sum_{i=t}^{n/2-1}2^{n/2-2-i} \binom{n/2-1}{i} \binom{i-t+1}{\floor{(i-t+1)/2}}. \nonumber 
	\end{align}
	After some algebraic manipulation of this expression, we conclude that the maximum number of redundant bits necessary is $\frac{1}{2} \log(n-2t)+2t+3$.
\end{proof}
\section{Symmetric Multiset Deletion-correcting Composition-Reconstruction Codes} \label{sec::sdel}
As mentioned in Section \ref{sec::error_models}, errors under this category occur in such a way that the affected multisets occur in pairs. We begin directly with the case when two symmetric multisets are inaccessible.
\begin{lemma}
	Consider a string $\boldsymbol{s} \in \mathcal{S}_{R}(n)$. Assume that for any $1 \leq k \leq \ceil{\frac{n-1}{2}}$, one is given $C'(\boldsymbol{s})=\bigcup_{i\in [n] \backslash \{k, n-k+1\}} C_i(\boldsymbol{s})$. Then, $\boldsymbol{s}$ can be fully recovered. \label{lem::sym_del}
\end{lemma}
\begin{proof}
	\begin{mycases}
		\case 
		$n$ is odd.
		
		Since the deleted multisets $C_k(\boldsymbol{s})$ and $C_{n-k+1}(\boldsymbol{s})$ can never be consecutive when $n$ is odd, we can infer from \cite[Lemma 4]{pattabiraman} that any attempt to substitute $C_{n-k+1}(\boldsymbol{s})$ with another multiset, say $C'_{n-k+1}$, that may or may not preserve the value of $\sigma_{k-1}(\boldsymbol{s})$, will surely cause a disagreement with $C_{n-k}(\boldsymbol{s})$. Hence, there exists no valid alternative choices for the multiset pair $\{C_k(\boldsymbol{s}), C_{n-k+1}(\boldsymbol{s})\}$, thus implying that $\boldsymbol{s}$ is uniquely reconstructable. 
		\case
		$n$ is even.
		
		As in the previous case, we can argue that for any $k\neq \{\frac{n}{2},\frac{n}{2}+1\}$, i.e., when the missing multisets are non-consecutive, $\boldsymbol{s}$ remains unique reconstructable by virtue of \cite[Lemma 4]{pattabiraman}. The only case left to be analyzed is when the deleted multisets are adjacent, i.e $C_{\frac{n}{2}}(\boldsymbol{s})$ and $C_{\frac{n}{2}+1}(\boldsymbol{s})$. More specifically, we examine the existence of any $\boldsymbol{v} \in \mathcal{S}_R(n)$, such that $$\bigcup_{i\in [n] \backslash \{\frac{n}{2}, \frac{n}{2}+1\}} C_i(\boldsymbol{v})=\bigcup_{i\in [n] \backslash \{\frac{n}{2}, \frac{n}{2}+1\}} C_i(\boldsymbol{s}).$$
		This directly leads to the following relations:
		\begin{equation}
			\begin{split}
			(\boldsymbol{s}_1^{n/2-2}, \boldsymbol{s}_{n/2+3}^n)&=(\boldsymbol{v}_1^{n/2-2}, \boldsymbol{v}_{n/2+3}^n), \\
			\sigma_i&=\sigma'_i, \quad \forall \;\; 1\leq i \leq \frac{n}{2}-2  \\
			\sigma_{\frac{n}{2}-1}+\sigma_{\frac{n}{2}}&=\sigma'_{\frac{n}{2}-1}+\sigma'_{\frac{n}{2}}. 
			\end{split} \nonumber
		\end{equation}
		where the sequence $\boldsymbol{\sigma}_{\boldsymbol{v}}=(\sigma'_1, \ldots, \sigma'_{n/2})$ describes $\boldsymbol{v}$.
		\subcase $\boldsymbol{\sigma}_{s}=\boldsymbol{\sigma}_{\boldsymbol{v}}$\\
		We only study this subcase for when $\sigma_{{\frac{n}{2}}-1}=\sigma'_{{\frac{n}{2}}-1}=1$ and $(s_{\frac{n}{2}-1}, s_{\frac{n}{2}+2}) \neq (v_{\frac{n}{2}-1}, v_{\frac{n}{2}+2})$, since the alternative involves $C_{n/2+1}(\boldsymbol{s})=C_{n/2+1}(\boldsymbol{v})$ and as a result of this, Lemma \ref{lem::del_set1} precludes the existence of $\boldsymbol{v}$, since $C(\boldsymbol{s})$ and $C(\boldsymbol{v})$ cannot differ by a single multiset alone. This situation is illustrated in Fig. \ref{fig::subcase1a}.
		
		We now proceed to ascertain if there exists some $\boldsymbol{v}$ for which $C_{n/2-1}(\boldsymbol{s})=C_{n/2-1}(\boldsymbol{v})$ holds. Alternatively, we need the following set equality relation to hold:
		\begin{equation}
		\left\{\begin{split}
		&\{c(\boldsymbol{s}_1^{\frac{n}{2}-2}),0\}\\
		&\{c(\boldsymbol{s}_2^{\frac{n}{2}-2}),0,s_+\}\\
		&\{c(\boldsymbol{s}_3^{\frac{n}{2}-2}),0,s_+, s_-\}\\
		&\{c(\boldsymbol{s}_{\frac{n}{2}+3}^n),1\}\\
		&\{c(\boldsymbol{s}_{\frac{n}{2}+3}^{n-1}),1,s_-\}\\
		&\{c(\boldsymbol{s}_{\frac{n}{2}+3}^{n-2}),1,s_+, s_-\}
		\end{split} \right\} = \left\{\begin{split}
		&\{c(\boldsymbol{v}_1^{\frac{n}{2}-2}),1\}\\
		&\{c(\boldsymbol{v}_2^{\frac{n}{2}-2}),1,v_{+}\}\\
		&\{c(\boldsymbol{v}_3^{\frac{n}{2}-2}),1,v_+, v_-\}\\
		&\{c(\boldsymbol{v}_{\frac{n}{2}+3}^n),0\}\\
		&\{c(\boldsymbol{v}_{\frac{n}{2}+3}^{n-1}),0,v_-\}\\
		&\{c(\boldsymbol{v}_{\frac{n}{2}+3}^{n-2}),0,v_+,v_-\}
		\end{split} \right\}. \label{eq::set_eq1aa}
		\end{equation}
		Due to the weight mismatch property between prefix and suffix of equal lengths, we note from Fig. \ref{fig::subcase1a} that if $\boldsymbol{v}$ must uphold:
		\begin{eqnarray}
			\mathrm{wt}(\boldsymbol{s}_2^{n/2-2})+1&<&\mathrm{wt}(\boldsymbol{s}_{n/2+3}^{n-1})\nonumber \\
			\implies \mathrm{wt}(\boldsymbol{s}_1^{n/2-2})+3&\leq&\mathrm{wt}(\boldsymbol{s}_{n/2+3}^{n}). \label{eq::lkop}
		\end{eqnarray}
		Now to prove that (\ref{eq::set_eq1aa}) never holds, it suffices to show that the composition $\{c(\boldsymbol{s}_{\frac{n}{2}+3}^n),1\}$ can never be matched to any two elements on the RHS in (\ref{eq::set_eq1aa}), even when (\ref{eq::lkop}) holds with equality. It is easy to see this when $v_++v_-<2$. On the contrary when $v_++v_-=2$, the compositions $\{c(\boldsymbol{v}_1^{\frac{n}{2}-2}),1\}$ and $\{c(\boldsymbol{v}_2^{\frac{n}{2}-2}),1,v_{+}\}$ become identical, and cannot be matched simultaneously to the components of RHS in (\ref{eq::set_eq1aa}). Therefore, $\boldsymbol{v}$ does not exist.
		\begin{figure}[!htb]
			\centering			
			\scalebox{0.8}{\begin{tikzpicture}[font=\ttfamily,
				array/.style={matrix of nodes,nodes={draw, minimum size=7mm, fill=green!30},column sep=-\pgflinewidth, row sep=0.5mm, nodes in empty cells,
					row 1/.style={nodes={draw=none, fill=none, minimum size=5mm}},
					row 1 column 1/.style={nodes={draw}}}]
				\draw (0,0) rectangle (2,0.6) node[pos=.5] {$\boldsymbol{s}_1^{\frac{n}{2}-2}$};
				
				\draw (2,0) rectangle (2.5,0.6) node[pos=0.5] {$0$};
				\draw (2.5,0) rectangle (3,0.6) node[pos=0.5] {$s_+$};
				\draw (3,0) rectangle (3.5,0.6) node[pos=0.5] {$s_{-}$};
				\draw (3.5,0) rectangle (4,0.6) node[pos=0.5] {$1$};
				

				\draw (4,0) rectangle (6,0.6) node[pos=0.5] {$\boldsymbol{s}_{\frac{n}{2}+3}^n$};
				
				\draw (0,-0.4) rectangle (2,-1) node[pos=0.5] {$\boldsymbol{v}_1^{\frac{n}{2}-2}$};
				
				\draw (2,-0.4) rectangle (2.5,-1) node[pos=0.5] {$1$};
				\draw (2.5,-0.4) rectangle (3,-1) node[pos=0.6] {$v_{+}$};
				\draw (3,-0.4) rectangle (3.5,-1) node[pos=0.6] {$v_{-}$};
				\draw (3.5,-0.4) rectangle (4,-1) node[pos=0.5] {$0$};
				\draw (4,-0.4) rectangle (6,-1) node[pos=0.5] {$\boldsymbol{v}_{\frac{n}{2}+3}^n$};
				\end{tikzpicture}}
			\caption{Strings $\boldsymbol{s}$ and $\boldsymbol{v}$ are such that $(\boldsymbol{s}_1^{\frac{n}{2}-2}, \boldsymbol{s}_{\frac{n}{2}+3}^n)=(\boldsymbol{v}_1^{\frac{n}{2}-2}, \boldsymbol{v}_{\frac{n}{2}+3}^n)$, where $v_{+}+v_{-}=s_{+}+s_{-}$.}
			\label{fig::subcase1a}
		\end{figure}
		\subcase $\boldsymbol{\sigma}_{s}\neq \boldsymbol{\sigma}_{\boldsymbol{v}}$\\
		All of the possible combinations of $(\sigma_{{\frac{n}{2}}-1},\sigma_{{\frac{n}{2}}})$ and $(\sigma'_{{\frac{n}{2}}-1},\sigma'_{{\frac{n}{2}}})$ that comprehensively cover this subcase are:
		\begin{itemize}
			\item $(\sigma_{{\frac{n}{2}}-1},\sigma_{{\frac{n}{2}}})=(1,2b)$ and $(\sigma'_{{\frac{n}{2}}-1},\sigma'_{{\frac{n}{2}}})=(2b,1)$.
			\item $(\sigma_{{\frac{n}{2}}-1},\sigma_{{\frac{n}{2}}})=(2,0)$ and $(\sigma'_{{\frac{n}{2}}-1},\sigma'_{{\frac{n}{2}}})=(1,1)$.
			\item $(\sigma_{{\frac{n}{2}}-1},\sigma_{{\frac{n}{2}}})=(0,2)$ and $(\sigma'_{{\frac{n}{2}}-1},\sigma'_{{\frac{n}{2}}})=(1,1)$.
		\end{itemize}
		where $b \in \mathbb{F}_2$. For the sake of brevity, we only prove the first instance. The remaining proofs run in a similar fashion.\\		
		To reiterate our objective, we check for the existence of a string $\boldsymbol{v} \in \mathcal{S}_R(n)$, for a given $\boldsymbol{s} \in \mathcal{S}_R(n)$, which are characterized as per the depiction in Fig. \ref{fig::subcase1}.
		\begin{figure}[!htb]
			\centering			
			\scalebox{0.8}{\begin{tikzpicture}[font=\ttfamily,
				array/.style={matrix of nodes,nodes={draw, minimum size=7mm, fill=green!30},column sep=-\pgflinewidth, row sep=0.5mm, nodes in empty cells,
					row 1/.style={nodes={draw=none, fill=none, minimum size=5mm}},
					row 1 column 1/.style={nodes={draw}}}]
				\draw (0,0) rectangle (2,0.6) node[pos=.5] {$\boldsymbol{s}_1^{\frac{n}{2}-2}$};
				
				\draw (2,0) rectangle (2.5,0.6) node[pos=0.5] {$s_+$};
				\draw (2.5,0) rectangle (3,0.6) node[pos=0.5] {$b$};
				\draw (3,0) rectangle (3.5,0.6) node[pos=0.5] {$b$};
				\draw (3.5,0) rectangle (4,0.6) node[pos=0.5] {$s_-$};	
				
				\draw (4,0) rectangle (6,0.6) node[pos=0.5] {$\boldsymbol{s}_{\frac{n}{2}+3}^n$};
				
				\draw (0,-0.4) rectangle (2,-1) node[pos=0.5] {$\boldsymbol{v}_1^{\frac{n}{2}-2}$};
				
				\draw (2,-0.4) rectangle (2.5,-1) node[pos=0.5] {$b$};
				\draw (2.5,-0.4) rectangle (3,-1) node[pos=0.6] {$v_{+}$};
				\draw (3,-0.4) rectangle (3.5,-1) node[pos=0.6] {$v_{-}$};
				\draw (3.5,-0.4) rectangle (4,-1) node[pos=0.5] {$b$};
				\draw (4,-0.4) rectangle (6,-1) node[pos=0.5] {$\boldsymbol{v}_{\frac{n}{2}+3}^n$};
				\end{tikzpicture}}
			\caption{Strings $\boldsymbol{s}$ and $\boldsymbol{v}$ are such that $(\boldsymbol{s}_1^{\frac{n}{2}-2}, \boldsymbol{s}_{\frac{n}{2}+3}^n)=(\boldsymbol{v}_1^{\frac{n}{2}-2}, \boldsymbol{v}_{\frac{n}{2}+3}^n)$, where $s_++s_-=v_{+}+v_-=1$.}
			\label{fig::subcase1}
		\end{figure}
		Since $\boldsymbol{s}$ and $\boldsymbol{v}$ may only differ in their respective composition multisets of substring lengths $\frac{n}{2}$ and $\frac{n}{2}+1$ alone, we endeavor to find the conditions that allow for the set equality of $C_{\frac{n}{2}-1}(\boldsymbol{s})$ and $C_{\frac{n}{2}-1}(\boldsymbol{v})$. More explicitly, we require:
		\begin{equation}
		\left\{\begin{split}
		&\{c(\boldsymbol{s}_1^{\frac{n}{2}-2}),s_+\}\\
		&\{c(\boldsymbol{s}_2^{\frac{n}{2}-2}),s_+,b\}\\
		&\{c(\boldsymbol{s}_3^{\frac{n}{2}-2}),s_+,b^2\}\\
		&\{c(\boldsymbol{s}_{\frac{n}{2}+3}^n),1-s_+\}\\
		&\{c(\boldsymbol{s}_{\frac{n}{2}+3}^{n-1}),1-s_+,b\}\\
		&\{c(\boldsymbol{s}_{\frac{n}{2}+3}^{n-2}),1-s_+,b^2\}
		\end{split} \right\} = \left\{\begin{split}
		&\{c(\boldsymbol{v}_1^{\frac{n}{2}-2}),b\}\\
		&\{c(\boldsymbol{v}_2^{\frac{n}{2}-2}),b,v_{+}\}\\
		&\{c(\boldsymbol{v}_3^{\frac{n}{2}-2}),b,01\}\\
		&\{c(\boldsymbol{v}_{\frac{n}{2}+3}^n),b\}\\
		&\{c(\boldsymbol{v}_{\frac{n}{2}+3}^{n-1}),b,1-v_+\}\\
		&\{c(\boldsymbol{v}_{\frac{n}{2}+3}^{n-2}),b,01\}
		\end{split} \right\}. \nonumber
		\end{equation}
		When $s_+=v_+=0$, we may proceed under the assumption that $\mathrm{wt}(\boldsymbol{s}_2^{n/2-2})=\mathrm{wt}(\boldsymbol{s}_{n/2+3}^{n-1})$ to account for the worst case. In this situation, either $\{c(\boldsymbol{s}_1^{\frac{n}{2}-2}),s_+\}$ or $\{c(\boldsymbol{s}_{\frac{n}{2}+3}^n),1-s_+\}$ fails to be matched, depending on the chosen value of $b$. Else when either $s_+$ or $v_+$ equals $1$, we infer that (\ref{eq::lkop}) holds true. Again, we choose to proceed with the worst case, i.e. $\mathrm{wt}(\boldsymbol{s}_2^{n/2-2})+3=\mathrm{wt}(\boldsymbol{s}_{n/2+3}^{n-1})$, and an exhaustive examination of each possibility reveals that the previous set equality cannot be satisfied. Thus, we conclude that $\boldsymbol{v}$ does not exist.
	\end{mycases}
\end{proof}
The previous result reveals that the codebook $\mathcal{S}_R(n)$ is sufficiently robust to correct the deletion of a single pair of symmetric multisets,i.e., 
\begin{theorem}
	The code $ \mathcal{S}_{R}(n)$ is a single symmetric multiset deletion correcting code. \label{lem::sym_del_new}
\end{theorem}
Consequently, if a single composition is substituted in $C(\boldsymbol{s})$ where $\boldsymbol{s} \in \mathcal{S}_R(n)$, then there occurs a mismatch between the cumulative weights of the specific multiset affected, say $C_i(\boldsymbol{s})$, and its symmetric counterpart $C_{n-i+1}(\boldsymbol{s})$. Now if both $C_i(\boldsymbol{s})$ and $C_{n-i+1}(\boldsymbol{s})$ are deleted, Lemma \ref{lem::sym_del} tells us that $\boldsymbol{s}$ is still uniquely recoverable. Thus, we conclude that $\mathcal{S}_R(n)$ is capable of correcting a single composition error just like $S_{CA}^{(1)}(n)$, as pointed out previously in Section \ref{sec::subs}.

We now investigate further along this direction and seek to determine if the absence of multiple pairs of such multisets impacts reconstructability. The deletion of two or more pairs of symmetric multisets, as shown in Lemma \ref{lem::25} (Appendix), no longer guarantees unique reconstruction of codewords drawn from $\mathcal{S}_R(n)$. To remedy this, we propose the code $\mathcal{S}^{(2)}_{DS}(n)$ [see Construction 6], capable of correcting deletions of two pairs of symmetric sets.

\begin{lemma}
	Consider a string  $\boldsymbol{s} \in \mathcal{S}^{(2)}_{DS}(n)$. Given only the composition multisets $\bigcup_{i\in [n] \backslash \{k-1,k,n-k+1,n-k+2\}} C_i(\boldsymbol{s})$, one can uniquely recover $\boldsymbol{s}$. \label{lem::sds2a}
\end{lemma}
\begin{proof}
	\begin{mycases}
		\case $n$ is even and the deleted multisets are neighboring, i.e. $\{C_{n/2-1}(\boldsymbol{s}), \ldots, C_{n/2+2}(\boldsymbol{s})\}$\\
		
		We recall from the proof of Lemma \ref{lem::25}, that for some $\boldsymbol{s} \in \mathcal{S}_R(n)$ characterized by $\boldsymbol{\sigma}_{s}=(\sigma_1, \ldots \sigma_{n/2})$, there may exist some $\boldsymbol{v} \in \mathcal{S}_R(n)$ with $\boldsymbol{\sigma}_{\boldsymbol{v}}=(\sigma'_1, \ldots, \sigma'_{n/2})$, such that:
		\begin{equation}
		\begin{split}
		\sigma_i&= \sigma'_i, \quad \quad \forall \;\;\; 1\leq i \leq \frac{n}{2}-3 \\
		\sigma_{\frac{n}{2}-2}+ \sigma_{\frac{n}{2}-1}+ \sigma_{\frac{n}{2}}&= \sigma'_{\frac{n}{2}-2}+ \sigma'_{\frac{n}{2}-1}+ \sigma'_{\frac{n}{2}}.
		\end{split}
		\label{uio}
		\end{equation}
		The difference of the sum of their respective cumulative weights for composition multisets containing substrings of lengths from $1$ to $\frac{n}{2}$, can be simplified to:
		\begin{eqnarray}
		&&\sum_{i=1}^{n/2} w_i(\boldsymbol{s})-\sum_{i=1}^{n/2} w_i(\boldsymbol{v}) \nonumber \\
		&=&\sum_{i=n/2-1}^{n/2} w_i(\boldsymbol{s})-\sum_{i=n/2-1}^{n/2} w_i(\boldsymbol{v}) \nonumber \\
		&=&3(\sigma'_{n/2-2}-\sigma_{n/2-2})+(\sigma'_{n/2-1}-\sigma_{n/2-1}).
		\end{eqnarray}
		The above difference is maximized when either:
		\begin{eqnarray}
		(\sigma_{\frac{n}{2}-2},\sigma_{\frac{n}{2}-1},\sigma_{\frac{n}{2}})&=&(0,1,2), \nonumber \\
		(\sigma'_{\frac{n}{2}-2},\sigma'_{\frac{n}{2}-1},\sigma'_{\frac{n}{2}})&=&(2,1,0). \nonumber
		\end{eqnarray}
		or:
		\begin{eqnarray}
		(\sigma_{\frac{n}{2}-2},\sigma_{\frac{n}{2}-1},\sigma_{\frac{n}{2}})&=&(0,2,2), \nonumber \\
		(\sigma'_{\frac{n}{2}-2},\sigma'_{\frac{n}{2}-1},\sigma'_{\frac{n}{2}})&=&(2,2,0). \nonumber
		\end{eqnarray}
		In either case, (\ref{uio}) is upheld. Hence we can write that:
		\begin{equation}
		\sum_{i=1}^{n/2} w_i(\boldsymbol{s})-\sum_{i=1}^{n/2} w_i(\boldsymbol{v})\leq 6. \nonumber
		\end{equation}
		\case $n$ may be odd/even and the deleted multisets are not all consecutive, i.e. $k+1<n-k+1$\\		
		From the proof of Lemma \ref{lem::25}, we note that when the multisets $\{C_{k-1}(\boldsymbol{s}), C_k(\boldsymbol{s}), C_{n-k+1}(\boldsymbol{s}),C_{n-k+2}(\boldsymbol{s})\}$ are deleted, there may exist an alternate $\boldsymbol{v} \in \mathcal{S}_R(n)$ such that: 
		\begin{eqnarray}
		\boldsymbol{s}_1^{k-3}&=&\boldsymbol{v}_1^{k-3}, \nonumber \\
		\boldsymbol{s}_{n-k+4}^n&=& \boldsymbol{v}_{n-k+4}^n, \nonumber \\
		\sigma_i&=& \sigma'_i, \quad \quad \forall \; i \in I\nonumber\\
		\sigma_k+2\sigma_{k-1}+3\sigma_{k-2}&=&\sigma'_k+2\sigma'_{k-1}+3\sigma'_{k-2}, \nonumber \\
		\sigma_{k+1}+\sigma_{k}+\sigma_{k-1}+\sigma_{k-2}&=&\sigma'_{k+1}+\sigma'_{k}+\sigma'_{k-1}+\sigma'_{k-2}.\nonumber
		\end{eqnarray}
		where $I=\big[\ceil{\frac{n}{2}}\big] \backslash \{k-2, \ldots, k+1\}$. As before, we bound the difference of the sum of cumulative weights of $\boldsymbol{s}$ and $\boldsymbol{v}$:
		\begin{eqnarray}
		\sum_{i=1}^{\ceil{n/2}} w_i(\boldsymbol{s})-\sum_{i=1}^{\ceil{n/2}} w_i(\boldsymbol{v}) 
		&=&\sum_{i=k-1}^{k} w_i(\boldsymbol{s})-\sum_{i=k-1}^{k} w_i(\boldsymbol{v}) \nonumber \\
		&=&(\sigma'_{k-1}-\sigma_{k-1})\nonumber\\
		&&+3(\sigma'_{k-2}-\sigma_{k-2}).
		\end{eqnarray}
		We find through numerical verification that this quantity cannot exceed 5, and it precisely occurs when:
		\begin{eqnarray}
		(\sigma_{k-2},\sigma_{k-1},\sigma_{k}, \sigma_{k+1})&=&(0,2,0,0), \nonumber \\
		(\sigma'_{k-2},\sigma'_{k-1},\sigma'_{k}, \sigma'_{k+1})&=&(1,0,1,0). \nonumber
		\end{eqnarray}
		
		As a result, in both cases the additional constraint $\sum_{i=1}^{\ceil{\frac{n}{2}}} w_i(\boldsymbol{s}) \text{ mod } 7=a$ in (\ref{eq::sds}) ensures unique reconstruction when the aforementioned multisets are lost.
	\end{mycases}
\vspace{-1em}
\end{proof}
The previous result permits us to conclude that
\begin{theorem}
The code $\mathcal{S}^{(2)}_{DS}(n)$ is a $2$-symmetric multiset deletion correcting code. 
\label{lem::sds2a_new}
\end{theorem}
We now seek to generalize the coding constraints in $\mathcal{S}^{(2)}_{DS}(m)$ in (\ref{eq::sds}) by examining how the required redundancy scales as more consecutive multiset pairs go missing. This is accomplished by $\mathcal{S}_{DS}^{\prime (t)}(n)$ [see Construction 7]. Theorem~\ref{lem::sdst_final} demonstrates that $\mathcal{S}_{DS}^{\prime  (t)}(n)$ is a {$t$-symmetric consecutive multiset deletion composition code}.
The proof commences with the following lemma.

\begin{lemma}
	Consider a string  $\boldsymbol{s} \in \mathcal{S}^{\prime (t)}_{DS}(n)$, where $t\geq 2$ and $n \geq 2t+4$. If one is given a corrupted composition multiset $C'(\boldsymbol{s})=\bigcup_{i\in [\ceil{n/2}] \backslash \{k-t, \ldots ,k-1\}} C_i(\boldsymbol{s}) \cup C_{n-i+1}(\boldsymbol{s})$  for any $t<k -1\leq \floor{n/2}$, i.e. $t$ consecutive symmetric multiset pairs are missing, $\boldsymbol{s}$ can be uniquely reconstructed. \label{lem::sdst}
\end{lemma}
\begin{proof}
	\begin{mycases}
		\case $n$ may be odd/even and the $2t$ deleted multisets are not adjacent, i.e. $k<n-k+2$.\\
		Since the multiset pairs $(C_i(\boldsymbol{s}), C_{n-i+1}(\boldsymbol{s}))$ have been eliminated, for $k-t\leq i \leq k-1$, we also do not know their respective cumulative weights. Thus, the values of $\sigma_{k-t-1}, \ldots \sigma_{k-2}$ are also unknown. Furthermore, we note from (\ref{eq::cum_wts}) that $\sigma_{k-1}$ and $\sigma_{k}$ are also not deducible. However, the sum of these missing values can be inferred from
		\begin{align}
			w_{k+1}-w_k&=(k+1)w_1-\sum_{i=1}^{k}i \sigma_{k+1-i} - kw_1+\sum_{i=1}^{k-1}i \sigma_{k-i} \nonumber \\
			&=	w_1 -\sigma_k\ldots -\sigma_1. \nonumber
		\end{align} 
		To test if $\boldsymbol{s}$ is uniquely recoverable, we attempt to find a suitable $\boldsymbol{v} \in \mathcal{S}_R(n)$, characterized by $\sigma'_1, \ldots ,\sigma'_{\ceil{n/2}}$, such that 
		\begin{equation}
		\widetilde{C}_i(\boldsymbol{s})=\widetilde{C}_i(\boldsymbol{v}). \quad \forall \;i \in \Big[\Big\lceil \frac{n}{2} \Big\rceil \Big]\backslash \{k-t,\ldots ,k-1\}  \nonumber
		\end{equation}
		These equations also imply that:
		\begin{align}
		(\boldsymbol{s}_1^{k-t-2},\boldsymbol{s}_{n-k+t+3}^n)&=(\boldsymbol{v}_1^{k-t-2},\boldsymbol{v}_{n-k+t+3}^n), \nonumber \\
		\sigma_i&= \sigma'_i, \quad \quad \forall \; i \in \Big[\Big\lceil \frac{n}{2} \Big\rceil \Big]\backslash I \nonumber \\
		\sum_{j \in I}\sigma_j&= \sum_{j \in I} \sigma'_j. \nonumber 
		\end{align}
		where $I=\{k-t-1, \ldots ,k\}$. Alike the approach undertaken in prior proofs, we now attempt to compute the maximum difference between the sum of cumulative weights of $\boldsymbol{s}$ and $\boldsymbol{v}$:
		\begin{align}
		\sum_{i=1}^{\ceil{n/2}} w_i(\boldsymbol{s})-w_i(\boldsymbol{v}) &= \sum_{i=k-t}^{k-1} w_i(\boldsymbol{s})-w_i(\boldsymbol{v}) \nonumber\\
		&= \sum_{i=k-t}^{k-1}\big(iw_1(\boldsymbol{s})-\sum_{j=1}^{i-1}j\sigma_{i-j}\big) \nonumber \\
		&-\sum_{i=k-t}^{k-1}\big(iw_1(\boldsymbol{v})-\sum_{j=1}^{i-1}j\sigma'_{i-j}\big) \nonumber \\
		&= \frac{t(t+1)}{2}(\sigma'_{k-t-1}-\sigma_{k-t-1}) \nonumber \\
		&+\ldots +3(\sigma'_{k-3}-\sigma_{k-3})\nonumber\\
		&+(\sigma'_{k-2}-\sigma_{k-2}). \label{eq::wdiff}
		\end{align}
		The final equality follows from $w_1(\boldsymbol{s})=w_1(\boldsymbol{v})$, which always holds since the premise of this error model states that $k-t>1$, suggesting that the multisets $C_1$ and $C_n$ are always preserved.
		\subcase $t$ is even.\\
		In this case, the quantity in (\ref{eq::wdiff}) is maximized when we have:
		\begin{align}
		(\sigma'_{k-t-1},  \ldots,  \sigma'_{k})&=(\overbrace{2, \ldots 2}^{\frac{t}{2}+1}, \overbrace{0, \ldots 0}^{\frac{t}{2}+1}), \nonumber \\
		(\sigma_{k-t-1},  \ldots,  \sigma_{k})&=(\overbrace{0, \ldots 0}^{\frac{t}{2}+1}, \overbrace{2, \ldots 2}^{\frac{t}{2}+1}). \nonumber
		\end{align}
		It is worth pointing out that these configurations may not always be valid, since the available multisets may not allow for them. However, they certainly embody the worst possible case. Now applying this to (\ref{eq::wdiff}), we obtain the following bound:
		\begin{eqnarray}
		\sum_{i=1}^{n/2} w_i(\boldsymbol{s})-w_i(\boldsymbol{v}) &\leq& \frac{t(t+2)^2}{4}. \label{eq::wdiff_even}
		\end{eqnarray}
		\subcase $t$ is odd.\\
		When $t$ is odd, the difference between the cumulative weights of $\boldsymbol{s}$ and $\boldsymbol{v}$ is maximized when:
		\begin{align}
	(\sigma'_{n/2-t},  \ldots,  \sigma'_{n/2})&=(\overbrace{2, \ldots 2}^{\frac{t-1}{2}},p, \overbrace{0, \ldots 0}^{\frac{t-1}{2}}), \nonumber \\
	(\sigma_{n/2-t},  \ldots,  \sigma_{n/2})&=(\overbrace{0, \ldots 0}^{\frac{t-1}{2}}, p,\overbrace{2, \ldots 2}^{\frac{t-1}{2}}). \nonumber
	\end{align}
		where $p \in \{0,1,2\}$. By further manipulating (\ref{eq::wdiff}), we get
		\begin{equation}
		\sum_{i=1}^{n/2} w_i(\boldsymbol{s})-w_i(\boldsymbol{v}) \leq \frac{t(t+1)(t+3)}{4}. \label{eq::wd_odd}
		\end{equation}
	\case $n$ is even and all of the deleted multisets are consecutive, i.e. $C_{n/2-t+1}(\boldsymbol{s}), \ldots, C_{n/2+t}(\boldsymbol{s})$.\\
	Much like the previous case, we attempt to find a $\boldsymbol{v} \in \mathcal{S}_R(n)$, characterized by $\sigma'_1, \ldots \sigma'_{{\frac{n}{2}}}$, such that for $1 \leq i \leq n/2-t$:
	\begin{equation}
	\widetilde{C}_i(\boldsymbol{s})=\widetilde{C}_i(\boldsymbol{v}), \nonumber 
	\end{equation}
	As a consequence, the following equalities also hold:
	\begin{align}
	(\boldsymbol{s}_1^{n/2-t-1},\boldsymbol{s}_{n/2+t+2}^n)&=(\boldsymbol{v}_1^{n/2-t-1},\boldsymbol{v}_{n/2+t+2}^n) \nonumber \\
	\sigma_i&= \sigma'_i, \quad \quad \forall \; i \in [n/2-t-1] \nonumber \\
	\sum_{j=n/2-t}^{n/2}\sigma_j&= \sum_{j=n/2-t}^{n/2}\sigma'_j. \nonumber 
	\end{align}
	Corresponding to (\ref{eq::wdiff}), we arrive at:
	\begin{align}
	\sum_{i=1}^{n/2} w_i(\boldsymbol{s})-w_i(\boldsymbol{v})&= \frac{t(t+1)}{2}(\sigma'_{\frac{n}{2}-t}-\sigma_{\frac{n}{2}-t})+\ldots \nonumber\\
	&+3(\sigma'_{\frac{n}{2}-2}-\sigma_{\frac{n}{2}-2}) +(\sigma'_{\frac{n}{2}-1}-\sigma_{\frac{n}{2}-1}). \nonumber
	\end{align}
	By appropriately assigning the vectors $(\sigma_{n/2-t}, \ldots, \sigma_{n/2})$ and $(\sigma'_{n/2-t}, \ldots, \sigma'_{n/2})$, we can upper-bound the preceding quantity as follows:
	\begin{equation}
		\sum_{i=1}^{n/2} w_i(\boldsymbol{s})-w_i(\boldsymbol{v}) \leq \begin{cases}
		\frac{(t+1)^3}{4},& \text{if $t$ is even.} \\
		\frac{t(t+1)(t+2)}{4},              & \text{otherwise.}
		\end{cases} \label{eq::wdiff2}
	\end{equation}
	\end{mycases}
	
	The definition of $\mathcal{S}^{\prime (t)}_{DS}(n)$ in (\ref{eq::cdst}) along with the bounds provided in (\ref{eq::wdiff_even}), (\ref{eq::wd_odd}) and (\ref{eq::wdiff2}) directly imply the statement.
\end{proof}
\begin{lemma}
	Consider a string  $\boldsymbol{s} \in \mathcal{S}^{\prime (t)}_{DS}(n)$, where $t\geq 2$ and $n\geq 2t+4$. If one is given a corrupted composition multiset $C'(\boldsymbol{s})=\bigcup_{i\in [\ceil{n/2}] \backslash [t]} C_i(\boldsymbol{s}) \cup C_{n-i+1}(\boldsymbol{s})$, $\boldsymbol{s}$ can be uniquely reconstructed. \label{lem::9}
\end{lemma}
\begin{proof}
	Unlike Lemma \ref{lem::sdst}, this proof is dedicated to the specific case where the multisets $C_1 \cup C_n, \ldots, C_t \cup C_{n-t+1}$ have been deleted. Since multisets $C_{t+1}(\boldsymbol{s})$ and $C_{t+2}(\boldsymbol{s})$ are available, we can obtain:
	\begin{eqnarray}
		w_{t+2}(\boldsymbol{s})-w_{t+1}(\boldsymbol{s})&=&w_1(\boldsymbol{s})-\sigma_{t+1}-\ldots-\sigma_1. \label{eq::sdstb1}
	\end{eqnarray}
	Similar to the prior analyses, we check for the existence of some $\boldsymbol{v} \in \mathcal{S}_R(n)$, specified by $\boldsymbol{\sigma}_{\boldsymbol{v}}=(\sigma'_1, \ldots, \sigma'_{\ceil{n/2}})$, that satisfies:
	\begin{equation}
		\widetilde{C}_i(\boldsymbol{s})=\widetilde{C}_i(\boldsymbol{v}). \label{eq::sdstb2}
	\end{equation}
	where $t<i\leq \ceil{n/2}$. From (\ref{eq::sdstb1}) and (\ref{eq::sdstb2}), we infer that for $1\leq i \leq \ceil{n/2}-t-1$:
	\begin{eqnarray}
		w_{t+i+1}(\boldsymbol{s})-w_{t+i}(\boldsymbol{s})&=&w_{t+i+1}(\boldsymbol{v})-w_{t+i}(\boldsymbol{v}) \nonumber\\
		\implies w_1(\boldsymbol{s})-\sigma_{t+i}-\ldots -\sigma_{1}&=&w_1(\boldsymbol{v})-\sigma'_{t+i}-\ldots -\sigma'_{1}. \nonumber
	\end{eqnarray}
	The preceding relation now allows us to deduce that:
	\begin{equation}
		\sigma_j=\sigma'_j.   \quad \forall \;\; t+2 \leq j \leq \ceil{n/2} \nonumber
	\end{equation}
	Also by construction of $\mathcal{S}_R(n)$, we observe that $\sigma_1=\sigma'_1$. As before, we inspect the difference of the sum of cumulative weights of $\boldsymbol{s}$ and $\boldsymbol{v}$:
	\begin{eqnarray}
	(\boldsymbol{s}_1^{k-t-2},\boldsymbol{s}_{n-k+t+3}^n)&=&(\boldsymbol{v}_1^{k-t-2},\boldsymbol{v}_{n-k+t+3}^n), \nonumber \\
	\sigma_i&=& \sigma'_i, \quad \quad \forall \; i \in \big[\ceil{n/2}\big]\backslash I \nonumber \\
	\sum_{j \in I}\sigma_j&=& \sum_{j \in I} \sigma'_j. \nonumber 
	\end{eqnarray}
	where $I=\{k-t-1, \ldots ,k\}$. The sum of cumulative weights of $\boldsymbol{s}$ and $\boldsymbol{v}$ differ by:
	\begin{eqnarray}
	\sum_{i=1}^{\ceil{n/2}} w_i(\boldsymbol{s})-w_i(\boldsymbol{v}) &=& \sum_{i=1}^{t} w_i(\boldsymbol{s})-w_i(\boldsymbol{v}) \nonumber\\
	&=& \sum_{i=1}^{t}\big(iw_1(\boldsymbol{s})-\sum_{j=1}^{i-1}j\sigma_{i-j}\big) \nonumber \\
	&&-\sum_{i=1}^{t}\big(iw_1(\boldsymbol{v})-\sum_{j=1}^{i-1}j\sigma'_{i-j}\big) \nonumber \\
	&=& \frac{t(t+1)}{2}(w_1(\boldsymbol{s})-w_1(\boldsymbol{v}))\nonumber\\
	&&+\frac{(t-2)(t-1)}{2}(\sigma'_{2}-\sigma_{2}) \nonumber \\
	&&+\ldots +3(\sigma'_{t-2}-\sigma_{t-2})\nonumber\\
	&&+(\sigma'_{t-1}-\sigma_{t-1}). \label{eq::wcdiff_sv}
	\end{eqnarray}
	Since $w_{t+1}(\boldsymbol{s})=w_{t+1}(\boldsymbol{v})$ and $w_{t+2}(\boldsymbol{s})=w_{t+2}(\boldsymbol{v})$, we rewrite (\ref{eq::sdstb1}) as:
	\begin{eqnarray}
		w_1(\boldsymbol{s})-w_1(\boldsymbol{v})&=&(\sigma_{t+1}-\sigma'_{t+1})+\ldots+(\sigma_{1}-\sigma'_{1}) \nonumber \\
		&=&(\sigma_{t+1}-\sigma'_{t+1})+\ldots+(\sigma_{2}-\sigma'_{2}) \nonumber\\
		&\leq& 2t. \label{eq::w1_ub}
	\end{eqnarray}
	We now attempt to design the vectors $\boldsymbol{\sigma}_{s}$ and $\boldsymbol{\sigma}_{\boldsymbol{v}}$ such that for a fixed value of $w_1(\boldsymbol{s})-w_1(\boldsymbol{v})$, the following quantity is maximized:
	\begin{equation}
		\frac{(t-2)(t-1)}{2}(\sigma'_{2}-\sigma_{2}) +\ldots +(\sigma'_{t-1}-\sigma_{t-1}). \nonumber
	\end{equation}
	while bearing in mind that:
	\begin{eqnarray}
		w_1(\boldsymbol{s})-w_1(\boldsymbol{v})&=&\sum_{i=2}^{t+1} (\sigma_i-\sigma'_i). \nonumber
	\end{eqnarray}
	Clearly, we must set $\sigma'_i-\sigma_i=2$ for $i=2, \ldots$, due to the higher weights of these terms, and $\sigma'_i-\sigma_i=-2$ for $i=t-1,t-2, \ldots$ on account of the minor influence of these terms on (\ref{eq::wcdiff_sv}). Additionally, we set $(\sigma_{t},\sigma'_{t})=(\sigma_{t+1},\sigma'_{t+1})=(2,0)$, thus allowing us to reduce the quantity $\sum_{i=2}^{t-1}(\sigma_i-\sigma'_i)$, i.e.
	\begin{equation}
		\sum_{i=2}^{t-1}(\sigma_i-\sigma'_i)=a-4. \nonumber
	\end{equation} 
	where $a=w_1(\boldsymbol{s})-w_1(\boldsymbol{v})$. Hence, to proceed with the maximization of (\ref{eq::wcdiff_sv}), we perform the following assignment when $a$ is odd:
	\begin{equation}
	\begin{split}
	(\sigma'_{2},  \ldots,  \sigma'_{t-1})&=({2, \ldots 2},p', {0, \ldots ,0}),  \\
	(\sigma_{2},  \ldots,  \sigma_{t-1})&=({0, \ldots 0},p, {2, \ldots ,2}). 
	\end{split}\label{eq::sig_assign}
	\end{equation}
	where $p+p'=1$. Here $p$ and $p'$ may be assigned interchangeably, depending on $t$. In a similar fashion, when $a$ is even, we again reuse this assignment while setting either $(p,p')=(0,2)$ or $p=p'=0$. Further noting that the term $w_1(\boldsymbol{s})-w_1(\boldsymbol{v})$ has the highest weight in (\ref{eq::wcdiff_sv}), we combine (\ref{eq::wcdiff_sv}), (\ref{eq::w1_ub}) and (\ref{eq::sig_assign}) to arrive at the following upper bound:
	\begin{eqnarray}
	\sum_{i=1}^{\ceil{n/2}} w_i(\boldsymbol{s})-w_i(\boldsymbol{v}) &\leq& 
	\Big \lceil\frac{4t^3}{3}+\frac{2t}{3}-\frac{35}{4}\Big \rceil. \label{eq::sdst_constraint}
	\end{eqnarray}
\end{proof}
Upon combining Lemmas \ref{lem::sdst} and \ref{lem::9}, we arrive at:
\begin{lemma}
	Consider a string  $\boldsymbol{s} \in \mathcal{S}^{\prime (t)}_{DS}(n)$ [see Construction~7], where $t\geq 2$ and $n \geq 2t+4$. If one is given a corrupted composition multiset $C'(\boldsymbol{s})=\bigcup_{i\in [\ceil{n/2}] \backslash \{k-t, \ldots ,k-1\}} C_i(\boldsymbol{s}) \cup C_{n-i+1}(\boldsymbol{s})$  for any $t \leq k -1\leq \floor{n/2}$, i.e. $t$ consecutive symmetric multiset pairs are missing, $\boldsymbol{s}$ can be uniquely reconstructed. \label{lem::sdst_final}
\end{lemma}
\begin{theorem}
$\mathcal{S}^{\prime (t)}_{DS}(n)$ is a $t$-symmetric consecutive multiset deletion composition code.
\end{theorem}
\emph{Remark:} Experimentally, it is found that an appropriate modulo constraint  corresponding to (\ref{eq::wdiff2}) is sufficient to allow the correction of deletion of any $t$ symmetric multiset pairs, consecutive or otherwise. An intuitive interpretation for this result follows from the fact that when the missing multiset pairs are consecutive, the least number of constraints are imposed on $\boldsymbol{\sigma}$. A rigorous proof for the same is yet to be found. It is also worth mentioning that though the constraint in (\ref{eq::sdst_constraint}) is stricter than that of (\ref{eq::wdiff2}), the order of the required redundancy remains identical.

\section{Skewed substitution-correcting codes} \label{sec::skew}
In this section, we confine our focus to the correction of skewed substitution errors [see Definition 6].

\begin{lemma}
	Consider any $\boldsymbol{s} \in \mathcal{S}_R(n)$. Given that there occurs a single skewed substitution error in its composition set, one can uniquely recover $\boldsymbol{s}$.
\end{lemma}
\begin{proof}
	In the following, we let the corrupted composition set be denoted by $C'(\boldsymbol{s})=\bigcup_{i \in [n]} C'_i(\boldsymbol{s})$.
	\begin{mycases}
		\case $n$ is even.\\
		Given $C'(\boldsymbol{s})$, it is easy to identify the corrupted composition multiset $C'_k(\boldsymbol{s})$, since the following relation only holds for $k$:
		\begin{equation}
		w'_k<w'_{n-k+1}.
		\end{equation}
		If we now delete all elements of $C'_k(\boldsymbol{s})$ from $C'(\boldsymbol{s})$, Lemma~\ref{lem::del_set1_new} tells us that $\boldsymbol{s}$ is still uniquely recoverable.
		\case $n$ is odd.\\
		Using the arguments of the preceding case, we can reach the same conclusion for an odd $n$, when the affected multiset is $C'_k(\boldsymbol{s})$, where $\ceil{n/2}<k\leq n$, because in these cases, there exists an uncorrupted distinct symmetric multiset $C'_{n-k+1}(\boldsymbol{s})$, which gives us the true cumulative weight and thus allows us to accurately recover $\boldsymbol{\sigma}_{s}$.\\
		If $k=\ceil{n/2}$, this is no longer true since the multiset $C_{\ceil{n/2}}(\boldsymbol{s})$ is its own symmetric counterpart. Noting that this normally helps us determine the bits $(s_{\ceil{n/2}-1}, s_{\ceil{n/2}+1})$, we recall from Lemma \ref{lem::comp_mismatch} that when these bits are assigned incorrectly, inconsistencies with the multiset $C_{\ceil{n/2}-1}$ would arise, which are not permitted under the considered error model. Hence, we conclude that $\boldsymbol{s}$ can be recovered uniquely.
	\end{mycases}
\end{proof}
We now consider a more general error model involving multiple asymmetric skewed substitution errors, wherein each multiset pair $\widetilde{C}_i$, for any $i \in [n]$, may contain at most one skewed substitution and the total number of errors does not exceed $t$. It is found that the asymmetric $t$-multiset deletion-correcting code $ \mathcal{S}^{(t)}_{DA}(n)$ is also robust to $t$ asymmetric skewed substitutions and in the following, we prove the same.
\begin{lemma}
	Consider any $\boldsymbol{s} \in \mathcal{S}^{(t)}_{DA}(n)$. Given that there occurs $t$ skewed asymmetric substitution errors in its composition set, such that for all $1\leq i \leq n$, $\widetilde{C}_i(\boldsymbol{s})$ contains at most one skewed substitution error, then one can uniquely recover $\boldsymbol{s}$. \label{lem::skewc}
\end{lemma}
\begin{proof}
	Since the error model only allows at most one skewed substitution in a pair of symmetric multisets, the cumulative weights of all sets can be determined accurately. This is due to the fact that if multiset $C_k(\boldsymbol{s})$ has been corrupted, we may write:
	\begin{equation}
		w_k < w_{n-k+1}. \label{eq::cwts}
	\end{equation} 
	As a consequence, all cumulative weights can be correctly re-assigned and in turn the $\boldsymbol{\sigma}_{s}$ sequence can be recovered. The preceding inequality also allows to identify the affected multisets, the deletion of which would transform our problem of correcting $t$ asymmetric skewed substitutions into reconstruction under the absence of $t$ multisets. According to Theorem~\ref{lem::t_dels}, unique reconstruction of $\boldsymbol{s}$ is perfectly possible, thus concluding our proof.
\end{proof}
The aforementioned result naturally leads to the following theorem.
\begin{theorem}
	$\mathcal{S}^{(t)}_{DA}(n)$ is a $t$-asymmetric skewed composition code. \label{lem::skewc}
\end{theorem}

\section{Conclusion} \label{sec::concl}
In this work, we propose and investigate error models involving insertion and deletion of substring compositions in the context of polymer-based data storage. In particular, we examine the robustness of the composition-reconstructable code introduced in \cite{pattabiraman, p2}, and identify the situations which do not guarantee unique reconstruction of codewords from this construction. For these cases, new codes are proposed. Notably, an equivalence between codes correcting multiset deletions and insertions is established. We also examine a special asymmetric variant of substitution errors, namely skewed substitution errors, which manifest in polymer-based storage.

Several problems pertaining to string construction under this data storage paradigm still remain open:
\begin{itemize}
	\item The error model involving skewed substitutions under a symmetric setting is yet to be investigated. It would be interesting to know if there exists a suitable codebook offering a lower redundancy than that designed to correct standard substitution errors under the symmetric setting, as stated in \cite{pattabiraman}.
	\item The problem of reconstructing strings from composition multisets, error-free or otherwise, could be extended to larger alphabets.
	\item Though some bounds on the maximum number of mutually equicomposable strings were stated in \cite{acharya}, bounds on the error ball sizes under the error models involving substitutions, insertions or deletions are still unknown. These could allow us to infer if the proposed code constructions are indeed optimal.
	\item One could also extend this research to the construct wherein bits are arranged in a circular fashion, on a ring.
	\item As pointed out in \cite{acharya}, a polynomial-time algorithm for the string reconstruction problem is yet to be found.
\end{itemize}


\appendix

\begin{lemma}
	Consider a string $\boldsymbol{s} \in \mathcal{S}_{R}(n)$. Given $C'(\boldsymbol{s})=\bigcup_{i\in [n] \backslash \{k-1,k, n-k+1, n-k+2\}} C_i(\boldsymbol{s})$ for any $1 \leq k < \ceil{\frac{n-1}{2}}$, $\boldsymbol{s}$ may no longer be uniquely determined. \label{lem::25}
\end{lemma}

\begin{proof}
	\begin{mycases}
		\case $n$ is even and deleted sets are: $\{C_{\frac{n}{2}-1}(\boldsymbol{s}), \ldots, C_{\frac{n}{2}+2}(\boldsymbol{s})\}$.
		
		To demonstrate that $\mathcal{S}_R(n)$ does not necessarily preserve unique reconstructability when the multisets $\{C_{\frac{n}{2}-1}, \ldots, C_{\frac{n}{2}+2}\}$ go missing, we consider two codewords $\boldsymbol{s}, \boldsymbol{v} \in \mathcal{S}_R(n)$, such that:
		\begin{equation}
			\bigcup_{i \in \{n, \ldots, \frac{n}{2}+3\}} C_i(\boldsymbol{s})=\bigcup_{i \in \{n, \ldots, \frac{n}{2}+3\}} C_i(\boldsymbol{v}).
		\end{equation}
		From our knowledge of the reconstruction algorithm [Section \ref{sec::prelim}], we can also infer the following:
		\begin{equation}
			\begin{split}
			(\boldsymbol{s}_1^{n/2-3}, \boldsymbol{s}_{n/2+4}^n)&=(\boldsymbol{v}_1^{n/2-3}, \boldsymbol{v}_{n/2+4}^n),  \\
			\sigma_i &= \sigma'_i. \quad \quad 1\leq i \leq \frac{n}{2}-3,\\
			\sigma_{\frac{n}{2}-2}+\sigma_{\frac{n}{2}-1}+\sigma_{\frac{n}{2}}&= \sigma'_{\frac{n}{2}-2}+\sigma'_{\frac{n}{2}-1}+\sigma'_{\frac{n}{2}}.
			\end{split} \label{eq::lkj}
		\end{equation}
		where $\boldsymbol{\sigma}_{s}=(\sigma_1, \ldots, \sigma_{n/2})$ and $\boldsymbol{\sigma}_{\boldsymbol{v}}=(\sigma'_1, \ldots, \sigma'_{n/2})$ correspond to $\boldsymbol{s}$ and $\boldsymbol{v}$ respectively. Additionally, we set:
		\begin{equation}
			\begin{split}
			(\sigma_{\frac{n}{2}-2},\sigma_{\frac{n}{2}-1},\sigma_{\frac{n}{2}})&=(0,0,1), \\
			(\sigma'_{\frac{n}{2}-2},\sigma'_{\frac{n}{2}-1},\sigma'_{\frac{n}{2}})&=(1,0,0), \\
			v_{n/2-2}&=1,\\
			s_{n/2}&=1,\\
			s_{n-3}&= 0, \\
			\mathrm{wt}(\boldsymbol{s}_2^{n/2-3})&=\mathrm{wt}(\boldsymbol{s}_{n/2+4}^{n-4}).
			\end{split} \label{eq::lkj2}
		\end{equation}
		The relations between $\boldsymbol{s}$ and $\boldsymbol{v}$ as described by (\ref{eq::lkj}) and (\ref{eq::lkj2}) are depicted in Fig. \ref{fig::mult_sets}. Evidently, $\boldsymbol{s}$ and $\boldsymbol{v}$ differ in their respective multisets $C_{n/2+2}$ and $C_{n/2+1}$ according Lemma \ref{lem::comp_mismatch}. Additionally, since their cumulative weights $w_{n/2+2}$ and $w_{n/2}$ also differ, as one may verify from (\ref{eq::cum_wts}) and (\ref{eq::lkj2}), we deduce that the multisets $C_{n/2}$ and $C_{n/2-1}$ also do not match for $\boldsymbol{s}$ and $\boldsymbol{v}$. We now proceed to examine if $C_{n/2-2}(\boldsymbol{s})=C_{n/2-2}(\boldsymbol{v})$ holds:
		
		\begin{equation}
		\left\{\begin{split}
		&\{c(\boldsymbol{s}_1^{\frac{n}{2}-3}),0\}\\
		&\{c(\boldsymbol{s}_2^{\frac{n}{2}-3}),0^2\}\\
		&\{c(\boldsymbol{s}_3^{\frac{n}{2}-3}),0^21\}\\
		&\{c(\boldsymbol{s}_4^{\frac{n}{2}-3}),0^31\}\\
		&\{c(\boldsymbol{s}_5^{\frac{n}{2}-3}),0^41\}\\
		&\{c(\boldsymbol{s}_{\frac{n}{2}+4}^n),0\}\\
		&\{c(\boldsymbol{s}_{\frac{n}{2}+4}^{n-1}),0^2\}\\
		&\{c(\boldsymbol{s}_{\frac{n}{2}+4}^{n-2}),0^3\}\\
		&\{c(\boldsymbol{s}_{\frac{n}{2}+4}^{n-3}),0^31\}\\
		&\{c(\boldsymbol{s}_{\frac{n}{2}+4}^{n-4}),0^41\}
		\end{split} \right\} = \left\{\begin{split}
		&\{c(\boldsymbol{v}_1^{\frac{n}{2}-3}), 1\}\\
		&\{c(\boldsymbol{v}_2^{\frac{n}{2}-3}),01\}\\
		&\{c(\boldsymbol{v}_3^{\frac{n}{2}-3}),0^21\}\\
		&\{c(\boldsymbol{v}_4^{\frac{n}{2}-3}),0^31\}\\
		&\{c(\boldsymbol{v}_5^{\frac{n}{2}-3}),0^41\}\\
		&\{c(\boldsymbol{v}_{\frac{n}{2}+4}^n),0\}\\
		&\{c(\boldsymbol{v}_{\frac{n}{2}+4}^{n-1}),0^2\}\\
		&\{c(\boldsymbol{v}_{\frac{n}{2}+4}^{n-2}),0^3\}\\
		&\{c(\boldsymbol{v}_{\frac{n}{2}+4}^{n-3}),0^4\}\\
		&\{c(\boldsymbol{v}_{\frac{n}{2}+4}^{n-4}),0^5\}
		\end{split} \right\}. \label{eq::mseteq}
		\end{equation}
		Using (\ref{eq::lkj2}) to simplify this set equality relation, we arrive at:
		
		\begin{equation}
		\left\{\begin{split}
		&\{c(\boldsymbol{s}_1^{\frac{n}{2}-3}),0\}\\
		&\{c(\boldsymbol{s}_2^{\frac{n}{2}-3}),0^2\}\\
		&\{c(\boldsymbol{s}_{\frac{n}{2}+4}^{n-3}),0^31\}\\
		&\{c(\boldsymbol{s}_{\frac{n}{2}+4}^{n-4}),0^41\}
		\end{split} \right\} = \left\{\begin{split}
		&\{c(\boldsymbol{v}_1^{\frac{n}{2}-3}),1\}\\
		&\{c(\boldsymbol{v}_2^{\frac{n}{2}-3}),01\}\\
		&\{c(\boldsymbol{v}_{\frac{n}{2}+4}^{n-3}),0^4\}\\
		&\{c(\boldsymbol{v}_{\frac{n}{2}+4}^{n-4}),0^5\}
		\end{split} \right\}.
		\end{equation}
		
		Since the construction of $\mathcal{S}_R(n)$ in () requires $s_1=0$ and (\ref{eq::lkj2}) mandates that $s_{n-3}=0$ and $\mathrm{wt}(\boldsymbol{s}_2^{n/2-3})=\mathrm{wt}(\boldsymbol{s}_{n/2+4}^{n-4})$, we are led to the following relation:
			\begin{equation}
		\mathrm{wt}(\boldsymbol{s}_1^{n/2-3})=\mathrm{wt}(\boldsymbol{s}_2^{n/2-3})=\mathrm{wt}(\boldsymbol{s}_{\frac{n}{2}+4}^{n-3})=\mathrm{wt}(\boldsymbol{s}_{\frac{n}{2}+4}^{n-4}).
		\end{equation}
		This allows us to conclude that (\ref{eq::mseteq}) indeed holds, and further bit specifications in $\boldsymbol{s}$ and $\boldsymbol{v}$ can lead us to similar set equality relations for the multisets $C_{n/2-3}, \ldots, C_1$. Hence, $\boldsymbol{s}$ and $\boldsymbol{v}$ become confusable under the deletion of multisets $\{C_{\frac{n}{2}-1}(\boldsymbol{s}), \ldots, C_{\frac{n}{2}+2}(\boldsymbol{s})\}$.
		
		\begin{figure}[!htb]
			\centering			
			\scalebox{0.8}{\begin{tikzpicture}[font=\ttfamily,
				array/.style={matrix of nodes,nodes={draw, minimum size=7mm, fill=green!30},column sep=-\pgflinewidth, row sep=0.5mm, nodes in empty cells,
					row 1/.style={nodes={draw=none, fill=none, minimum size=5mm}},
					row 1 column 1/.style={nodes={draw}}}]
				\draw (0,0) rectangle (2,0.6) node[pos=.5] {$\boldsymbol{s}_1^{\frac{n}{2}-3}$};
				\draw (2,0) rectangle (2.5,0.6) node[pos=0.5] {$0$};
				\draw (2.5,0) rectangle (3,0.6) node[pos=0.5] {$0$};
				\draw (3,0) rectangle (3.5,0.6) node[pos=0.5] {$1$};
				\draw (3.5,0) rectangle (4,0.6) node[pos=0.5] {$0$};
				\draw (4,0) rectangle (4.5,0.6) node[pos=0.5] {$0$};
				\draw (4.5,0) rectangle (5,0.6) node[pos=0.5] {$0$};
				\draw (5,0) rectangle (7,0.6) node[pos=0.5] {$\boldsymbol{s}_{\frac{n}{2}+4}^n$};
				
				\draw (0,-0.2) rectangle (2,-0.8) node[pos=0.5] {$\boldsymbol{v}_1^{\frac{n}{2}-3}$};
				\draw (2,-0.2) rectangle (2.5,-0.8) node[pos=0.6] {$1$};
				\draw (2.5,-0.2) rectangle (3,-0.8) node[pos=0.6] {$0$};
				\draw (3,-0.2) rectangle (3.5,-0.8) node[pos=0.6] {$0$};
				\draw (3.5,-0.2) rectangle (4,-0.8) node[pos=0.6] {$0$};
				\draw (4,-0.2) rectangle (4.5,-0.8) node[pos=0.6] {$0$};
				\draw (4.5,-0.2) rectangle (5,-0.8) node[pos=0.6] {$0$};
				
				\draw (5,-0.2) rectangle (7,-0.8) node[pos=0.5] {$\boldsymbol{v}_{\frac{n}{2}+4}^n$};
				\end{tikzpicture}}
			\caption{Strings $\boldsymbol{s}$ and $\boldsymbol{v}$ are specified by (\ref{eq::lkj}) and (\ref{eq::lkj2}).}
			\label{fig::mult_sets}
		\end{figure}
	
		\case
		$n$ may be odd/even and the four deleted sets are not consecutive: $\{C_{k-1}(\boldsymbol{s}), C_k(\boldsymbol{s}), C_{n-k+1}(\boldsymbol{s}),C_{n-k+2}(\boldsymbol{s})\}$, where $k+1<n-k+1$.
		
		In the following, we once again proceed by checking if $\boldsymbol{s}$ is uniquely recoverable, by probing the existence of some $\boldsymbol{v} \in \mathcal{S}_R(n)$, characterized by $\sigma'_1, \ldots, \sigma'_{\ceil{\frac{n}{2}}}$ such that for all $i \in [n] \backslash \{k-1,k-n-k+1,n-k+2\}$:
		\begin{equation}
		C_i(\boldsymbol{s})=  C_i(\boldsymbol{v}). \label{eq::civ}
		\end{equation}
		
		\subcase $k=2$
		
		This situation corresponds to the deletion of multisets $C_1(\boldsymbol{s})$, $C_2(\boldsymbol{s})$, $C_{n-1}(\boldsymbol{s})$ and $C_{n}(\boldsymbol{s})$. When this happens, for any $3\leq i \leq \ceil{n/2}-1$, the following values are recoverable:
		\begin{equation}
		w_{i+1}(\boldsymbol{s})-w_i(\boldsymbol{s})=\sigma_{i+1}+\ldots+\sigma_{\ceil{n/2}}. \nonumber
		\end{equation}
		This can be used to recover the values of $\sigma_4, \ldots, \sigma_{\ceil{n/2}}$. In other words,
		\begin{equation}
		\sigma_i=\sigma'_i. \quad  \forall \;\; 4 \leq i \leq \ceil{n/2} \label{eq::sigis}
		\end{equation}
		Furthermore, since $w_3(\boldsymbol{s})=w_3(\boldsymbol{v})$, we can infer from (\ref{eq::wk_alt}) and (\ref{eq::sigis}) that:
		\begin{eqnarray}
		\sigma_1+2\sigma_2+3\sigma_3&=&\sigma'_1+2\sigma'_2+3\sigma'_3 \nonumber \\
		\implies 2\sigma_2+3\sigma_3&=&2\sigma'_2+3\sigma'_3. \nonumber
		\end{eqnarray}
		The second equality follows from the construction of $\mathcal{S}_R(n)$. Given the above relation, we conclude that (\ref{eq::sigis}) also holds for $i \in \{2,3\}$. Moreover, we cannot have $(s_2,s_{n-1}) \neq (v_2,v_{n-1})$ even when $\sigma_2=\sigma'_2=1$, since the Catalan-Bertrand structure would automatically imply that $(s_2,s_{n-1}) = (v_2,v_{n-1})=(0,1)$. This inference combined with Lemma \ref{lem::comp_mismatch}, lead us to the conclusion that no suitable $\boldsymbol{v}$  exists.
		
		\subcase $k=3$
		
		When multisets $C_2(\boldsymbol{s}), C_3(\boldsymbol{s}), C_{n-2}(\boldsymbol{s})$ and $C_{n-1}(\boldsymbol{s})$ have been deleted, the availability of cumulative weights $w_1,w_4, \ldots w_{\ceil{n/2}}$ allow us to retrieve $\sigma_1, \sigma_5,\ldots, \sigma_{\ceil{n/2}}$ as in the previous subcase, i.e. 
		\begin{equation}
		\sigma_i=\sigma'_i. \quad  \forall \;\; i \in \big[\ceil{n/2}\big] \backslash \{2,3,4\}  \label{eq::sigis2}
		\end{equation}
		We also observe from (\ref{eq::cum_wts}) and (\ref{eq::civ}) that:
		\begin{eqnarray}
		w_4(\boldsymbol{s})-w_1(\boldsymbol{s})&=&w_4(\boldsymbol{v})-w_1(\boldsymbol{v}) \nonumber\\
		&=&3w_1(\boldsymbol{s})-\sigma_3-2\sigma_{2}-3\sigma_1,\nonumber\\
		\implies \sigma_2+2\sigma_3	&=& \sigma'_2+2\sigma'_3. \label{eq::sigi23}
		\end{eqnarray}
		Similarly, since $w_5(\boldsymbol{s})=w_5(\boldsymbol{v})$, we obtain:
		\begin{equation}
		\sigma_2+2\sigma_3+3\sigma_4=\sigma'_2+2\sigma'_3+3\sigma_4. \nonumber
		\end{equation}
		As a consequence, (\ref{eq::sigis2}) also holds for $i=4$. This, along with (\ref{eq::w1}) hint that:
		\begin{equation}
		\sigma_2+\sigma_3=\sigma'_2+\sigma'_3. \label{eq::sig23}
		\end{equation}
		Equations (\ref{eq::sigi23}) and (\ref{eq::sig23}) together insinuate that $(\sigma_2,\sigma_3)=(\sigma'_2,\sigma'_3)$. Hence, we may argue as before, that no suitable $\boldsymbol{v}$ distinct from $\boldsymbol{s}$ actually exists.\\
	
		\subcase $k\geq 4$

		Similar to the approach used in Case 1, we attempt to show that there exist two codewords $\boldsymbol{s}, \boldsymbol{v} \in \mathcal{S}_R(n)$, such that for all $i \in [n]\backslash \{k-1,k,n-k+1,n-k+2\}$:
		\begin{equation}
			C_i(\boldsymbol{s})= C_i(\boldsymbol{v}). 
		\end{equation}
		To this end, we construct a specific pair of strings $\boldsymbol{s}$ and $\boldsymbol{v}$ as follows:
		\begin{equation}
			\begin{split}
			(\boldsymbol{s}_1^{k-3}, \boldsymbol{s}_{n-k+4}^{n})&= (\boldsymbol{v}_1^{k-3}, \boldsymbol{v}_{n-k+4}^{n}), \\
			(\sigma_{k-2},\sigma_{k-1}, \sigma_{k}, \sigma_{k+1})&=(1,1,1,0),\\
			(\sigma'_{k-2},\sigma'_{k-1}, \sigma'_{k}, \sigma'_{k+1})&=(2,0,0,1),\\
			\sigma_i&=\sigma'_i, \quad \quad \forall \; k+2 \leq i \leq \ceil{\frac{n}{2}}\\
			(s_{k-1}, s_{k}, s_{k+1}, s_{k+2})&=(0,0,1),\\
			s_2&=1, \\
			v_{k-2}&=0.
			\end{split} \label{eq::lkopp}
		\end{equation}
		
		\begin{figure}[!htb]
			\centering			
			\scalebox{0.8}{\begin{tikzpicture}[font=\ttfamily,
				array/.style={matrix of nodes,nodes={draw, minimum size=7mm, fill=green!30},column sep=-\pgflinewidth, row sep=0.5mm, nodes in empty cells,
					row 1/.style={nodes={draw=none, fill=none, minimum size=5mm}},
					row 1 column 1/.style={nodes={draw}}}]
				\draw (0,0) rectangle (2,0.6) node[pos=.5] {$\boldsymbol{s}_1^{k-3}$};
				\draw (2,0) rectangle (2.5,0.6) node[pos=0.5] {$0$};
				\draw (2.5,0) rectangle (3,0.6) node[pos=0.5] {$0$};
				\draw (3,0) rectangle (3.5,0.6) node[pos=0.5] {$0$};
				\draw (3.5,0) rectangle (4,0.6) node[pos=0.5] {$1$};
				\draw (4,0) rectangle (5.5,0.6) node[pos=0.5] {$\boldsymbol{s}_{k+2}^{n-k-1}$};
				\draw (5.5,0) rectangle (6,0.6) node[pos=0.5] {$0$};
				\draw (6,0) rectangle (6.5,0.6) node[pos=0.5] {$1$};
				\draw (6.5,0) rectangle (7,0.6) node[pos=0.5] {$1$};
				\draw (7,0) rectangle (7.5,0.6) node[pos=0.5] {$0$};
				\draw (7.5,0) rectangle (9.5,0.6) node[pos=0.5] {$\boldsymbol{s}_{n-k+4}^n$};
				
				\draw (0,-0.2) rectangle (2,-0.8) node[pos=0.5] {$\boldsymbol{v}_1^{k-3}$};
				\draw (2,-0.2) rectangle (2.5,-0.8) node[pos=0.6] {$0$};
				\draw (2.5,-0.2) rectangle (3,-0.8) node[pos=0.6] {$0$};
				\draw (3,-0.2) rectangle (3.5,-0.8) node[pos=0.6] {$0$};
				\draw (3.5,-0.2) rectangle (4,-0.8) node[pos=0.6] {$1$};
				\draw (4,-0.2) rectangle (5.5,-0.8) node[pos=0.6] {$\boldsymbol{v}_{k+2}^{n-k-1}$};
				\draw (5.5,-0.2) rectangle (6,-0.8) node[pos=0.6] {$1$};
				\draw (6,-0.2) rectangle (6.5,-0.8) node[pos=0.6] {$0$};
				\draw (6.5,-0.2) rectangle (7,-0.8) node[pos=0.6] {$0$};
				\draw (7,-0.2) rectangle (7.5,-0.8) node[pos=0.6] {$1$};
				\draw (7.5,-0.2) rectangle (9.5,-0.8) node[pos=0.5] {$\boldsymbol{v}_{n-k+4}^n$};
				\end{tikzpicture}}
			\caption{Strings $\boldsymbol{s}$ and $\boldsymbol{v}$ are related such that $(\boldsymbol{s}_1^{k-3}, \boldsymbol{s}_{n-k+4}^n)=(\boldsymbol{v}_1^{k-3}, \boldsymbol{v}_{n-k+4}^n)$ and $c(\boldsymbol{s}_{k+2}^{n-k-1})=c(\boldsymbol{v}_{k+2}^{n-k-1})$}
			\label{fig::mult_sets2}
		\end{figure}
	
		These relations have been illustrated in Fig. \ref{fig::mult_sets2}. The preceding equalities also imply that:
		\begin{equation}
			\begin{split}
			\sigma_i&=\sigma'_i, \quad \quad \forall \; 1 \leq i \leq k-3 \\
			\sum_{i=k-2}^{k+1}\sigma_{i} &= \sum_{i=k-2}^{k+1} \sigma'_{i},\\
			\sigma_k+2\sigma_{k-1}+3\sigma_{k-2}&=\sigma'_k+2\sigma'_{k-1}+3\sigma'_{k-2},\\
			c(\boldsymbol{s}_{k+2}^{n-k-1})&= c(\boldsymbol{v}_{k+2}^{n-k-1}).
			\end{split}
		\end{equation}
		In turn, these relations help ensure that:
		\begin{equation}
			\begin{split}
			w_i(\boldsymbol{s})&=w_i(\boldsymbol{v}), \quad\quad \forall \; 1 \leq i \leq k-2 \\
			w_{k+1}(\boldsymbol{s})-w_{k-2}(\boldsymbol{s})&= w_{k+1}(\boldsymbol{v})-w_{k-2}(\boldsymbol{v}), \\
			w_{k+i+1}(\boldsymbol{s})-w_{k+i}(\boldsymbol{s})&=w_{k+i+1}(\boldsymbol{v})-w_{k+i}(\boldsymbol{v}).
			\end{split}
		\end{equation}
		for $1\leq i \leq n-k-1$. One may verify this with the assistance of (\ref{eq::w1}) and (\ref{eq::cum_wts}). \\
		From Fig. \ref{fig::mult_sets2}, it is fairly evident that $\boldsymbol{s}$ and $\boldsymbol{v}$ do not match in their corresponding multisets $C_{n-k+2}$ and $C_{n-k+1}$. Now as done in case 1, we check if multisets $C_{n-k}(\boldsymbol{s})$ and $C_{n-k}(\boldsymbol{v})$ match:
		\begin{equation}
		\left\{\begin{split}
		&\{c(\boldsymbol{s}_1^{k-3}),0^41,c\}\\
		&\{c(\boldsymbol{s}_2^{k-3}),0^41^2,c\}\\
		&\{c(\boldsymbol{s}_3^{k-3}),0^41^3,c\}\\
		&\{c(\boldsymbol{s}_{n-k+4}^n),0^21^3,c\}\\
		&\{c(\boldsymbol{s}_{n-k+4}^{n-1}),0^31^3,c\}\\
		&\{c(\boldsymbol{s}_{n-k+4}^{n-2}),0^41^3,c\}
		\end{split} \right\} = \left\{\begin{split}
		&\{c(\boldsymbol{v}_1^{k-3}),0^31^2,c\}\\
		&\{c(\boldsymbol{v}_2^{k-3}),0^41^2,c\}\\
		&\{c(\boldsymbol{v}_3^{k-3}),0^51^2,c\}\\
		&\{c(\boldsymbol{v}_{n-k+4}^n),0^21^3,c\}\\
		&\{c(\boldsymbol{v}_{n-k+4}^{n-1}),0^31^3,c\}\\
		&\{c(\boldsymbol{v}_{n-k+4}^{n-2}),0^41^3,c\}
		\end{split} \right\}.
		\end{equation}
		where $c=c(\boldsymbol{s}_{k+2}^{n-k-1})=c(\boldsymbol{v}_{k+2}^{n-k-1})$.
		By applying (\ref{eq::lkopp}) to this, we deduce that this equality is indeed upheld, thus implying that $\boldsymbol{s}$ and $\boldsymbol{v}$ are confusable under the absence of multisets $C_{k-1}, C_k, C_{n-k+1}, C_{n-k+2}$. 	
	\end{mycases}
\end{proof}

\end{document}